\documentclass{article}
\usepackage{microtype}
\usepackage{graphicx}
\usepackage{subcaption}
\usepackage{booktabs}
\usepackage{comment}
\usepackage{hyperref}

\usepackage{amsmath}
\usepackage{amssymb}
\usepackage{mathtools}
\usepackage{amsthm}
\usepackage[capitalize,noabbrev]{cleveref}
\theoremstyle{plain}
\newtheorem{theorem}{Theorem}[section]
\newtheorem{proposition}[theorem]{Proposition}
\newtheorem{lemma}[theorem]{Lemma}
\newtheorem{corollary}[theorem]{Corollary}
\theoremstyle{definition}
\newtheorem{definition}[theorem]{Definition}
\newtheorem{assumption}[theorem]{Assumption}
\theoremstyle{remark}

\theoremstyle{definition}
\newtheorem{example}{Example}[section]
\usepackage[textsize=tiny]{todonotes}

\newcommand{\E}{\mathbb{E}}
\newcommand{\Prob}{\mathbb{P}}
\newcommand{\Var}{\mathrm{Var}}

\newcommand{\cH}{\mathcal{H}}

\newcommand{\ip}[2]{\left\langle #1,#2\right\rangle}

\definecolor{mydarkblue}{rgb}{0,0.08,0.45}
\hypersetup{
    colorlinks=true,
    linkcolor=mydarkblue,
    citecolor=mydarkblue,
    filecolor=mydarkblue,
    urlcolor=mydarkblue,
    pdfview=FitH
}
\usepackage[main, preprint]{neurips_2026}
\usepackage[utf8]{inputenc}
\usepackage[T1]{fontenc}    
\usepackage{hyperref}      
\usepackage{url}            
\usepackage{booktabs}       
\usepackage{amsfonts}       
\usepackage{nicefrac}       
\usepackage{microtype}      
\usepackage{xcolor}        
\usepackage{algorithmic}
\usepackage{algorithm}

\title{An Asymptotic Analysis of the Shapley Value for Dataset Valuation}

\author{Mélissa Tamine \thanks{Corresponding author.} \\
Criteo AI lab, FairPlay joint team, France \\
CREST, ENSAE, Institut Polytechnique de Paris \\
\texttt{m.tamine@criteo.com} \\
\And
Benjamin Heymann \\
Criteo AI lab, FairPlay joint team, France \\
\texttt{b.heymann@criteo.com} \\
\AND
Maxime Vono \\
Criteo AI lab, FairPlay joint team, France \\
\texttt{m.vono@criteo.com} \\
\And
Patrick Loiseau \\
Inria, FairPlay joint team, France \\
\texttt{patrick.loiseau@inria.fr}
}

\begin{document}

\maketitle

\begin{abstract}
We propose an asymptotic analysis of the Shapley value in a \emph{dataset valuation} setting in which utilities are modeled as smooth functionals of empirical distributions via reproducing kernel Hilbert space (RKHS) mean embeddings. We prove that, despite its combinatorial definition, the Shapley value of a data source is asymptotically captured by a simple leading term. This term can be interpreted as the first-order contribution of a dataset relative to the surrounding data population. It also identifies the scale of the Shapley value as the number of data sources grows and provides a framework for analyzing existing Shapley value estimators. Moreover, for practitioners working with large numbers of datasets, the leading term becomes a tractable reference against which Shapley value approximations can be benchmarked.
\end{abstract}
\section{Introduction}
\label{sec:introduction}

Achieving strong generalization in machine learning often requires more data than any single party can access on its own. In many applications, the relevant data are naturally distributed across several contributors, each holding only a partial and potentially complementary view of the learning problem at hand. In online advertising, for instance, different retailers may hold complementary information about users’ browsing and purchasing behavior. Taken separately, these datasets provide only a partial picture, but taken together, they may support significantly better models. Such collaborative settings raise a natural question: if several contributors jointly enable the success of a learning task, how should that value be attributed back to their respective datasets? This is the \emph{dataset valuation} problem \citep{agarwal2019, tay2022, sim2020}. It generalizes the more popular \emph{data valuation} problem \cite{sim2022, jiang2023}, with data valuation being the special case of dataset valuation in which each owner contributes exactly one data point. A natural way to formalize it is through cooperative game theory: each dataset owner is modeled as a player, each coalition of players is assigned a utility reflecting the performance achieved by pooling their data, and a valuation rule is then used to distribute this utility among contributors \citep{ghorbani2019, jia2019b}. Among the solution concepts considered for this purpose, the Shapley value \citep{shapley1953} has become the most prominent in machine learning because it enjoys a strong axiomatic justification.
\\ \\
A major challenge in using the Shapley value is its computational intractability. For a fixed game with $I$ players, the Shapley value of one player averages its marginal contribution over all subsets of the other $I-1$ players. Computing it exactly, therefore, requires evaluating exponentially many coalitions. This has motivated a substantial literature on scalable approximations in the context of dataset valuation \citep{liu2022, wang2024a, cai2024, chi2025, wang2024b, sun2026, tamine2025}, including permutation-based Monte Carlo methods \citep{maleki2013, ghorbani2019, jia2019b}, group-testing-based accelerations \citep{jia2019b}, and dataset-level proxies such as DU-Shapley \citep{garrido2024}. These works address the problem of Shapley estimation for a fixed number of players.
\\ \\
Our paper is orthogonal to this fixed-game perspective and, to our knowledge, provides the first asymptotic analysis of the Shapley value in a dataset valuation setting. The main result is a finite-$I$ bound which, interpreted asymptotically, shows that the Shapley value is close to a simple leading term. This leading term depends on just a few interpretable quantities. So even though the Shapley formula is combinatorial, asymptotically, the value of a fixed dataset boils down to those interpretable quantities. This asymptotic characterization has two main outcomes. First, it identifies the scale of the Shapley value. Second, the leading term becomes an asymptotic reference: since the exact Shapley value is provably close to it in large‑$I$ regimes, any Shapley estimator can also be analyzed asymptotically by its distance to the same term, and at large scale, can even be benchmarked against it when the exact Shapley value computation becomes infeasible.

Conducting this asymptotic analysis, however, requires specifying what it means for the number of players in a cooperative game to grow, since the Shapley value is defined only for a fixed finite set of players, and there is no canonical way to let the player count tend to infinity. We resolve this by fixing a single data owner $i$, keeping its dataset $D_i$ unchanged, and constructing larger games by adding surrounding data owners sampled independently from a common population. The Shapley value of owner $i$ in the fixed game of $I$ players is denoted $\phi_i^I$, and the asymptotic object we study in this paper is the sequence $(\phi_i^I)_{I \ge 1}$. 

We also assume that the utility of a dataset depends smoothly on its empirical distribution. We represent empirical distributions through kernel mean embeddings in a reproducing kernel Hilbert space (RKHS). In this representation, large coalitions of surrounding owners concentrate around a deterministic population reference embedding. Adding the fixed dataset $D_i$ to a large coalition of datasets can then be viewed as a perturbation around this population reference. The smoothness of the utility allows us to analyze this perturbation by a first-order Taylor expansion, which reveals the leading term described above.

The rest of the paper develops this analysis. We summarize our main contributions as follows.

\begin{itemize}
    \item[(a)] We prove that the Shapley value $\phi_i^I$ is $O(1/I)$-close in $L^1$ to an interpretable leading term $\Theta_i^I$ (Theorem \ref{thm:leading-term}). This leading term depends, in particular, on a first-order signal that captures how player $i$'s dataset differs from the surrounding data population in directions relevant to the utility. As a consequence, under a non-degeneracy condition on this signal, we show that $\phi_i^I$ has scale $\log I/I$ in probability.
    \item[(b)] Using the scale identified by the leading term, we derive a general absolute-to-relative error principle: any estimator whose absolute error is $o(\log I/I)$ is asymptotically consistent in relative error. We apply this criterion to permutation-based Monte Carlo \cite{maleki2013} and group-testing-based estimators \citep{jia2019b}, subject to suitable computational budgets. We also introduce a structural class of leading-order estimators, defined by their closeness to the leading term, and show that it includes DU-Shapley \cite{garrido2024} and stratified Monte Carlo Shapley \cite{maleki2013}.
    \item[(c)] We complement the theory with an experimental benchmarking based on the leading term. Specifically, in large-scale regimes, we use the leading term as a reference against which Shapley value approximations can be benchmarked.
\end{itemize}

\paragraph{Related works.} The study of asymptotic properties of the Shapley value has a long tradition in cooperative game theory. Classical work on nonatomic games defines values on a continuum of infinitesimal players \citep{aumann1974}, while asymptotic-value theory studies limits of finite games induced by measure spaces of players \citep{neyman1981,neyman1988}. Our setting is different: each data owner remains a finite, non-negligible player, and coalition utilities are induced by empirical data distributions rather than by an abstract measure-theoretic game. More recent data valuation works share our goal of looking beyond a single fixed game. However, they replace the finite‑game Shapley value with a population quantity \cite{ghorbani2020, kwon2021efficient}, or treat it as random and study its sampling distribution for a fixed number of players \cite{wu2024, jia2026}. In contrast, we keep the Shapley definition unchanged and analyze its behavior as the number of players increases. Specifically, our object of interest is the sequence $(\phi_i^I)_{I\ge 1}$ of exact Shapley values of a fixed data owner. An extended discussion of related work is deferred to Appendix~\ref{sec:extended-related-work}.

\paragraph{Notation.} Throughout the paper, $\|\cdot\|$ denotes the norm of the ambient Hilbert space (or the Euclidean norm when applied to finite-dimensional vectors). For a sequence of random variables $(X_I)_{I\ge1}$ and a deterministic positive sequence $(a_I)_{I\ge1}$, we write $X_I=O_{\mathbb P}(a_I)$ if $X_I/a_I$ is bounded in probability, and $X_I=o_{\mathbb P}(a_I)$ if
$X_I/a_I \xrightarrow{\mathbb P}0$. Additional notations are along the main text and summarized in Table~\ref{tab:notation}.

\section{Problem Setup}
\label{sec:setup}

We consider a collaborative machine learning setting in which several data owners, henceforth called \emph{players}, pool their datasets to solve a common learning task. 

\subsection{Dataset valuation as a cooperative game}
\label{subsec:dataset-valuation-game}

Let $\mathcal Z$ be a measurable data space. In supervised learning, one may think of $\mathcal Z=\mathcal X\times\mathcal Y$, where $\mathcal X$ is the feature space and $\mathcal Y$ is the label space. However, in this paper, we keep the notation general because the analysis only uses the distributional content of the datasets.

Let $\mathcal{I}$ be a finite set of players. Each player $i \in \mathcal{I}$ owns a finite multiset $D_i=\{z_{i,1},\dots,z_{i,n_i}\}\in\mathcal M(\mathcal Z),$ with $n_i:=|D_i|\ge 1,$ where $\mathcal M(\mathcal Z)$ denotes the set of finite multisets over $\mathcal Z$. For a set of players $S$, we write $D_S:=\biguplus_{i\in S}D_i$ for the pooled dataset and $n_S:=|D_S|=\sum_{i\in S}n_i$ for its total number of datapoints. Here, $\biguplus$ denotes multiset union.

A dataset valuation problem is specified by a utility function $v:\mathcal M(\mathcal Z)\to\mathbb R^{+},$ which assigns a score to any pooled dataset. For the finite set of players $\mathcal I$, this induces the cooperative game
\begin{align*}
    (\mathcal I,u),
    \qquad
    u(S):=v(D_S),
    \quad S\subseteq\mathcal I.
\end{align*}
Thus, the value of a coalition is the utility obtained by pooling the datasets of its members.
\subsection{Shapley value}
\label{subsec:shapley-value}
The Shapley value of player $i\in\mathcal I$ is the average marginal contribution of $D_i$ over all possible coalitions of other players. In the finite $(\mathcal{I},u)$ game, it is
\begin{align}
    \phi_i(u)
    =
    \frac{1}{|\mathcal I|}
    \sum_{S\subseteq \mathcal I\setminus\{i\}}
    \binom{|\mathcal I|-1}{|S|}^{-1}
    \left[
        v(D_S\biguplus D_i)-v(D_S)
    \right].
    \label{eq:finite-shapley}
\end{align}
This is the finite-game object whose asymptotic behavior we study.

\subsection{From a single game to a sequence of games}
\label{subsec:growing-player-regime}

As explained in Section~\ref{sec:introduction}, an asymptotic analysis of the Shapley value requires specifying how games with different numbers of players are related.

To this end, we fix a data owner, denoted by $i$. Its dataset $D_i$ is kept fixed throughout the analysis; in particular, neither $D_i$ nor $n_i:=|D_i|$ depends on the number of players in the game. The remaining players represent the surrounding population of data sources.

For each $I\ge 2$, we consider an $I$-player game consisting of the fixed owner $i$ and $I-1$ surrounding owners. We index the surrounding owners by $[I-1]:=\{1,\dots,I-1\},$ and denote their datasets by $D_1,\dots,D_{I-1}$. Thus, the player set in this game is $\mathcal I_I := \{i\}\cup [I-1]$ where the symbol $i$ is kept distinct from the surrounding-player labels. For each $I$, the game is $(\mathcal{I}_I,u_I),$ with $u_I(S):=v(D_S)$ for all $S\subseteq \mathcal{I}_I.$

We denote by $\phi_i^I$ the Shapley value of the fixed player $i$ in this $I$-player game. Grouping coalitions by their cardinality in Eq.~\eqref{eq:finite-shapley} gives the following formulation of $\phi_i^I$
\begin{align}
    \phi_i^I
    =
    \frac{1}{I}
    \sum_{k=0}^{I-1}
    \mathbb E_{S_k}
    \left[
        v(D_{S_k}\biguplus D_i)-v(D_{S_k})
    \right],
    \label{eq:shapley-I-cardinality}
\end{align}
where $S_k$ is sampled uniformly among all subsets of $[I]\setminus\{i\}$ of size $k$. The asymptotic object we study in this paper is the sequence $(\phi_i^I)_{I \ge 1}$. 

\subsection{Smooth utilities of empirical embeddings}
\label{subsec:embedding-utility-model}
We now specify the class of utilities $v$ used in the asymptotic analysis. To analyze how utility changes when the fixed dataset $D_i$ is added to a large coalition, we choose a representation in which datasets can be compared via their empirical distributions, and in which small changes in these distributions lead to controlled changes in utility. We use kernel mean embeddings for this purpose. They map the empirical distribution of a dataset to an element of a reproducing kernel Hilbert space (RKHS), so that adding a dataset corresponds to perturbing an empirical mean in that space. This gives a convenient setting for applying Taylor expansions to the utility.

Let $(\mathcal H,\langle\cdot,\cdot\rangle)$ be a RKHS on $\mathcal Z$, with feature map $\varphi:\mathcal Z\to\mathcal H.$ For a nonempty finite multiset $D=\{z_1,\dots,z_n\}$, let $\widehat P_D:=\frac{1}{n}\sum_{r=1}^n \delta_{z_r}]$ be its empirical distribution. We define its empirical mean embedding by $\mu(\widehat P_D):= \frac{1}{n}\sum_{r=1}^n \varphi(z_r)$. For the empty dataset, we use the convention $\mu(\widehat P_\emptyset)=0.$ 

For the asymptotic analysis, we impose two regularity properties. Assumption \ref{ass:bounded-feature} ensures empirical embeddings of finite datasets cannot have an arbitrarily large norm. Assumption \ref{ass:smooth-utility} ensures that when a dataset is added to a large coalition, the resulting change in utility can be controlled by a first-order Taylor expansion.
\begin{assumption}[Bounded feature map]
\label{ass:bounded-feature}
There exists $\kappa<\infty$ such that $\|\varphi(z)\|\le \kappa, \forall z\in\mathcal Z.$
\end{assumption}

\begin{assumption}[Smooth embedding utility]
\label{ass:smooth-utility}
There exists a Fréchet differentiable functional $F:\mathcal H\to\mathbb R$
such that $v(D)=F\!\left(\mu(\widehat P_D)\right)$ for every finite multiset $D$. Moreover, there exist positive constants $G,M<\infty$ such that, for all $x,y\in\mathcal H$, $\|\nabla F(x)\|\le G,$ and $\|\nabla F(x)-\nabla F(y)\|\le M\|x-y\|$.
\end{assumption}

\subsection{Population model for surrounding players}
\label{subsec:population-model}
We now specify how the surrounding players other than $i$ are generated. This is the part of the setup that gives meaning to the growing-player regime: as $I$ increases, new players are added by sampling new datasets from the same population of data sources.

\begin{assumption}[Population of surrounding players]
\label{ass:population}
The surrounding players are independent draws from a common population of data sources. Formally, for each surrounding player $j$, let $\mathsf D_j$ be a random finite multiset with distribution $\mathcal P$ on $\mathcal M(\mathcal Z)$, and let $N_j := |\mathsf D_j|$ be its random dataset size. We assume that there exists $n_{\max}<\infty$ such that $1 \le N_j \le n_{\max}$ almost surely. We denote the average dataset size in the surrounding population by $\bar n := \mathbb E[N_j] \in [1,n_{\max}]$. After sampling, we denote by $D_j$ the realized dataset and by $n_j: = |D_j|$ its realized size.
\end{assumption}

Assumption~\ref{ass:population} is stated at the level of whole datasets. It does not require all players to have the same datapoint distribution. Rather, it says that each surrounding player is an independent draw from the same population of data sources. Thus, the model is exchangeable at the player level, while still allowing heterogeneity across players.

\begin{example}[Latent-type population]
\label{ex:latent-type-population}

A special case of Assumption~\ref{ass:population} is a hierarchical population model in which data owners belong to different latent subpopulations. Let $\mathcal T$ denote the set of possible subpopulation types and let $\Pi$ be a distribution on $\mathcal T$. For each
$j\neq i$, draw a latent type $\tau_j\sim\Pi$. Conditional on $\tau_j$, draw a dataset size $n_j\in\{1,\dots,n_{\max}\}$, possibly from a type dependent distribution, and then draw $z_{j,1},\dots,z_{j,n_j} \stackrel{\mathrm{i.i.d.}}{\sim} P_{\tau_j},$ where $P_\tau$ is a distribution on $\mathcal Z$. This construction satisfies Assumption~\ref{ass:population}: the datasets $D_j$ are independent and identically distributed, while the latent types allow different players to represent different parts of the population.
\end{example}

We now identify the embedding around which large coalitions of surrounding players concentrate. For a surrounding player $j\neq i$, define $X_j:=\sum_{z\in D_j}\varphi(z)\in\mathcal H$. This is the total contribution of player $j$'s dataset in the RKHS feature space. By Assumptions~\ref{ass:bounded-feature} and~\ref{ass:population}, $\|X_j\|\le n_{\max}\kappa$, so $\mathbb E[X_j]$ is well-defined.

For a nonempty coalition $S\subseteq[I]\setminus\{i\}$, the empirical embedding of its pooled dataset is $\mu(\widehat P_{D_S})= \frac{1}{n_S}\sum_{j\in S}X_j,$ where $n_S=\sum_{j\in S}n_j$. The numerator is the sum of the RKHS embedding of all datapoints pooled by the players in S, while the denominator is the total number of datapoints contributed by those players. By the law of large numbers, when $S$ contains many surrounding players, $\frac{1}{\lvert S \rvert}\sum_{j\in S}X_j$ is close to $\mathbb E[X_j]$ and $\frac{1}{\lvert S \rvert} n_S$ is close to $\bar n$. Hence, the empirical embedding of a large surrounding coalition is close to $\mu^\star:=\frac{\mathbb E[X_j]}{\bar n} =\frac{1}{\bar n}\mathbb E\left[\sum_{z\in D_j}\varphi(z)\right].$

The embedding $\mu^\star$ is the typical datapoint-level embedding of the surrounding population. Equivalently, it is the embedding approached by the pooled empirical distribution of a large coalition of independently sampled surrounding players. The next section uses the concentration of $\mu(\widehat P_{D_S})$ around $\mu^\star$ to identify the leading term governing the asymptotic behavior of $\phi_i^I$.
\section{Asymptotic Characterization of the Shapley Value}
\label{sec:asymptotic-characterization}

We now characterize the Shapley value of the fixed player $i$ as the number of surrounding players grows. The setup of Section~\ref{sec:setup} identifies the key intuition of the analysis: the empirical embedding of a large coalition of surrounding players concentrates around the reference embedding $\mu^\star$. This concentration allows us to approximate the marginal contribution of $D_i$ by a first-order expansion around $\mu^{\star}$. When averaged through the Shapley formula, this first-order approximation yields an explicit leading term, denoted by $\Theta_i^I$, which governs the large-I behavior of $\phi_i^I$. The purpose of this section is to explain how this concentration turns the Shapley formula into this interpretable \emph{leading term} asymptotically.

Recall from \eqref{eq:shapley-I-cardinality} that the Shapley value is an average over coalitions of players. To understand $\phi_i^I$, it is therefore enough to understand the marginal contribution of fixed player $i$ to a typical coalition $S$ of surrounding players. Let $\Delta_i(S) := v(D_S\biguplus D_i)-v(D_S)$ be this marginal contribution. For a nonempty coalition $S\subseteq[I]\setminus\{i\}$, write $\mu_S:=\mu(\widehat P_{D_S})$, and $\mu_i:=\mu(\widehat P_{D_i})$. Since $D_i$ is fixed throughout the asymptotic regime, both $n_i$ and $\mu_i$ are deterministic.

The first observation is geometric. Adding $D_i$ to $D_S$ changes the
coalition embedding by $\mu_{S\cup\{i\}}-\mu_S = \frac{n_i}{n_S+n_i}(\mu_i-\mu_S)$.
This identity is proved in Lemma~\ref{lem:update} of Appendix \ref{sec:technical-lemmas}. It shows that adding $D_i$ perturbs the current coalition embedding in the direction $\mu_i-\mu_S$, and the magnitude of the perturbation is the relative size of $D_i$ inside the enlarged coalition.

Since the utility has the form $v(D)=F(\mu(\widehat P_D)),$ the marginal contribution is the change of $F$ under this perturbation. Lemma~\ref{lem:first-order-marginal-expansion} in Appendix \ref{sec:technical-lemmas} gives the corresponding first-order expansion:
\begin{align}
\Delta_i(S) = \frac{n_i}{n_S+n_i} \left\langle \nabla F(\mu_S),\mu_i-\mu_S\right\rangle+R_S,
\label{eq:marginal-first-order}
\end{align}
where the Taylor remainder $R_S$ is controlled uniformly. Thus, at the level of a single coalition, the marginal contribution is governed by the local sensitivity of the utility at $\mu_S$ and by the direction in which adding $D_i$ moves the empirical embedding.

The asymptotic simplification comes from the population model. If $|S|=k$ is
large, then $S$ is representative of the surrounding population: Lemma~\ref{lem:coalition-embedding-concentration} in Appendix~\ref{sec:technical-lemmas} shows that $\mathbb E\|\mu_S-\mu^\star\|= O(k^{-1/2}).$
At the same time, the total number of datapoints in $S$ is close to its typical value $k\bar n$. Therefore, the population-level first-order approximation of $\Delta_i(S)$ is obtained from \eqref{eq:marginal-first-order} by replacing $\mu_S$ with $\mu^\star$ and $n_S$ with $k\bar n$, i.e, $\Delta_i(S) = \frac{n_i}{k\bar n} \left\langle \nabla F(\mu^\star), \mu_i-\mu^\star \right\rangle + \epsilon_S$
where $\epsilon_S$ collects the error made by this population-level replacement together with the Taylor remainder $R_S$.

This motivates the definition of the first-order signal.

\begin{definition}[First-order signal]
\label{def:first-order-signal}
The first-order signal of the fixed player $i$ is $c_i := \left\langle \nabla F(\mu^\star),\mu_i-\mu^\star\right\rangle$.
\end{definition}

The scalar $c_i$ measures the part of the difference between $D_i$ and the surrounding population that matters to the utility at first order. The vector $\mu_i-\mu^\star$ describes how the empirical distribution of $D_i$ differs from the population reference. The gradient $\nabla F(\mu^\star)$ describes how the utility changes locally around that reference. Their inner product keeps exactly the component of the distributional difference that changes the utility to first order.

For a coalition of size $k$, the resulting population-level marginal contribution is therefore $\frac{n_i}{k\bar n}c_i$. The factor $n_i/(k\bar n)$ is the relative size of the fixed dataset compared with a typical coalition of $k$ surrounding players. Now, the Shapley value averages marginal contributions uniformly over coalition sizes $k=0,\dots,I-1$. Therefore, averaging the population-level first-order contribution over $k=1,\dots,I-1$ gives the leading term
\begin{align*}
    \Theta_i^I:=\frac{1}{I}\sum_{k=1}^{I-1}\frac{n_i}{k\bar n}c_i
    =\frac{n_i c_i}{\bar n}
    \frac{H_{I-1}}{I},
    \qquad
    H_{I-1}:=\sum_{k=1}^{I-1}\frac{1}{k}.
\end{align*}
The empty-coalition contribution does not appear in $\Theta_i^I$, since it is weighted by $1/I$ in the Shapley formula, and is absorbed into the final $O(1/I)$ error given in Theorem \ref{thm:leading-term}.

The leading term separates three effects:
\begin{align*}
    \Theta_i^I
    =
    \underbrace{n_i}_{\text{dataset size}}
    \times
    \underbrace{c_i}_{\text{first-order signal}}
    \times
    \underbrace{\frac{H_{I-1}}{I\bar n}}_{\text{Shapley averaging over sizes}}.
\end{align*}
Thus, the asymptotic value of $D_i$ is driven by how much data the fixed player contributes, how this data differs from the surrounding population in a direction that affects the utility, and how the Shapley formula averages marginal effects over coalition sizes. Theorem \ref{thm:leading-term} states that this leading term captures the exact Shapley value asymptotically.

\begin{theorem}[Leading-term approximation]
\label{thm:leading-term}
Suppose Assumptions~\ref{ass:bounded-feature},
\ref{ass:smooth-utility}, and~\ref{ass:population} hold. Then, for the fixed
player $i$, $\mathbb E
    \left[
        \left|
            \phi_i^I-\Theta_i^I
        \right|
    \right]
    =
    O\!\left(\frac{1}{I}\right),$ where the expectation is over the surrounding players
$(D_j)_{j\neq i}$.
\end{theorem}

The proof is given in Appendix~\ref{subsec:proof-thm-leading-term}. Hence, although $\phi_i^I$ is defined through exponentially many coalitions, its asymptotic behavior is governed by the explicit quantity $\Theta_i^I$.
This approximation also identifies the order of magnitude of $\phi_i^I$ when the first-order signal is nonzero. Indeed, by Lemma~\ref{lem:harmonic-asymptotics} in Appendix~\ref{sec:technical-lemmas}, $H_{I-1}=\log I+O(1)$, and therefore $\Theta_i^I=\frac{n_i c_i}{\bar n}\frac{\log I}{I}+O\!\left(\frac{1}{I}\right)$. If $c_i\neq0$, the leading term is of order $\log I/I$. Moreover, Theorem~\ref{thm:leading-term} shows that
$\phi_i^I-\Theta_i^I=O_{L^1}\!\left(\frac{1}{I}\right),$ which is negligible compared with $\log I/I$. Thus, the exact Shapley value inherits the same asymptotic scale.

\begin{corollary}[Asymptotic scale of the Shapley value]
\label{cor:shapley-lower-bound}
Under the assumptions of Theorem~\ref{thm:leading-term}, if $c_i\neq 0$, then $|\phi_i^I|=\Theta_{\mathbb P}\!\left(\frac{\log I}{I}\right).$
\end{corollary}

The proof is given in Appendix~\ref{subsec:proof-cor-shapley-lower-bound}. The result shows that, under the non-degeneracy condition $c_i\neq0$, the exact Shapley value vanishes at order $\log I/I$. This order of magnitude determines how approximation errors should be compared to $\phi_i^I$. In particular, an absolute error that tends to zero is not sufficient: since $\phi_i^I$ itself goes to zero, such an error may still be large relative to the value being estimated. The meaningful requirement is that the absolute error be $o(\log I/I)$, i.e., smaller than the order of magnitude of $\phi_i^I$. This yields the following absolute-to-relative error criterion.

\begin{corollary}[Absolute error implies relative error]
\label{cor:absolute-to-relative}
Assume the conditions of Theorem~\ref{thm:leading-term} and suppose
$c_i\neq 0$. Let $\widehat\phi_i^I$ be any estimator of $\phi_i^I$ such that $\left|\widehat\phi_i^I-\phi_i^I\right|=o_{\mathbb P}\!\left(\frac{\log I}{I}\right)$. Then $\left|\frac{\widehat\phi_i^I}{\phi_i^I}
-1\right|\xrightarrow{\mathbb P}0.$
\end{corollary}

The proof is given in Appendix~\ref{subsec:proof-cor-absolute-to-relative}. Corollary~\ref{cor:absolute-to-relative} connects the asymptotic characterization of the exact Shapley value to the asymptotic study of estimators. In fact, any approximation can now be evaluated through relative error. The next section applies this criterion to several Shapley value estimators.
\section{Implications for Shapley Estimators}
\label{sec:estimators}
Section~\ref{sec:asymptotic-characterization} establishes two conclusions about the Shapley value of the fixed player $i$. First, the exact Shapley value is close to the leading term (Theorem~\ref{thm:leading-term}). Second, when $c_i\neq0$, its order of magnitude is $\log I/I$. The key point for approximation is that errors must be compared to this $\log I/I$ scale. An estimator is relatively consistent as soon as its absolute error is $o_{\mathbb P}(\log I/I)$, by Corollary~\ref{cor:absolute-to-relative}. This section applies this principle to two types of estimators.
\paragraph{Budget-based Shapley estimators.}
We first consider estimators that depend on a computational budget. The role of our asymptotic analysis is to translate this budget into a relative-error guarantee. By Corollary~\ref{cor:absolute-to-relative}, it is enough to choose the budget so that the absolute estimation error is $o_{\mathbb P}(\log I/I)$. The first estimation is the permutation Monte Carlo Shapley \citep{jia2019b, ghorbani2019}, denoted by $\widehat\phi_{i,\mathrm{MC}}^I$, which averages $m_I$ independent random marginal contributions of player $i$. Its formal definition is recalled in Appendix~\ref{sec:missing-defs}. The second is the group-testing Shapley estimator of \citep{jia2019b}, denoted by $\widehat\phi_{i,\mathrm{GT}}^I$, which estimates the Shapley value from $T_I$ coalition queries. Its definition is also given in Appendix~\ref{sec:missing-defs}.

\begin{proposition}[$\widehat{\phi}_{i,\mathrm{MC}}^I$ and $\widehat{\phi}_{i,\mathrm{GT}}^I$ have converging relative error under suited budgets]
\label{prop:mc_gt_budgets}
Let $\widehat\varphi_{i,\mathrm{MC}}^I$ be the permutation Monte Carlo Shapley estimator \cite{maleki2013} based on $m_I$ independent random permutations. Let $\widehat\varphi_{i,\mathrm{GT}}^I$ be the group-testing Shapley estimator of \cite{jia2019b}. Under the assumptions of Theorem~\ref{thm:leading-term} and with $c_i \neq 0$:
\begin{enumerate}
\item If $m_I \gg \frac{I^2}{\log^2 I}$, then $\left|\frac{\widehat{\phi}_{i,\mathrm{MC}}^I}{\phi_i^I} - 1\right| \xlongrightarrow{\mathbb{P}} 0$.
\item If there exist sequences $\varepsilon_I=o(\log I/I)$ and $\delta_I\to0$ such that the group-testing budget $T_I$ is large enough for the estimator to be an
$(\varepsilon_I,\delta_I)$-approximation of the Shapley vector $(\phi_j^I)_{j=1}^{I}$ in $\ell_2$-norm., then $\left|\frac{\widehat{\phi}_{i,\mathrm{GT}}^I}{\phi_i^I} - 1\right| \xlongrightarrow{\mathbb{P}} 0$.
\end{enumerate}
\end{proposition}

The proof is given in Appendix~\ref{subsec:proof-prop-mc-gt}.

\paragraph{Leading-order Shapley estimators.}
We then consider estimators that directly follow the structure of the leading term $\Theta_i^I = \frac{1}{I}\sum_{k=1}^{I-1}\frac{n_i}{k\bar n}c_i$. This term averages a per‑size contribution $\frac{n_i}{k\bar n}c_i$ over all coalition sizes $k$. An estimator that, for each $k$, approximates this per‑size contribution well enough so that the average error is $O(1/I)$ will also be $O(1/I)$-close to $\phi_i^I$ (by Theorem~\ref{thm:leading-term}) and hence relatively consistent. The following definition formalizes this.

\begin{definition}[Leading‑order estimator]
\label{def:leading_order}
An estimator $\widehat{\phi}_i^I$ is called a leading‑order estimator if it can be written as $\widehat{\phi}_i^I = \frac{1}{I}\sum_{k=1}^{I-1} \frac{n_i}{k\bar{n}}\,\widehat{\Delta}_{k,I}^{(i)} + \frac{\widehat{\Delta}_{0,I}^{(i)}}{I},$
where:
\begin{itemize}
\item[(C1)] $\mathbb{E}|\widehat{\Delta}_{0,I}^{(i)}| \le C_0$ for some constant $C_0$ independent of $I$;
\item[(C2)] $\displaystyle\sum_{k=1}^{I-1}\frac{1}{k}\,\mathbb{E}\bigl|\widehat{\Delta}_{k,I}^{(i)} - c_i\bigr| \le C_\Delta$ for some constant $C_\Delta$ independent of $I$.
\end{itemize}
\end{definition}

Condition (C1) controls the empty-coalition term, whose Shapley weight is $1/I$. Condition (C2) controls the cumulative error made when replacing the first-order signal $c_i$ by the estimator-specific quantities $\widehat{\Delta}_{k,I}^{(i)}$. The factor $1/k$ is the same harmonic weight that appears in the leading term, so the condition matches the structure of $\Theta_i^I$.

\begin{theorem}[Relative consistency of leading‑order estimators]
\label{thm:leading_order_consistent}
Under Assumptions~\ref{ass:population}--\ref{ass:smooth-utility} and with $c_i \neq 0$, every leading‑order estimator satisfies $\left|\frac{\widehat{\phi}_i^I}{\phi_i^I} - 1\right| \xlongrightarrow{\mathbb{P}} 0.$
\end{theorem}

The proof is given in Appendix~\ref{subsec:proof-thm-leading-order-estimator}. We now instantiate the leading-order definition with two estimators. The first is DU-Shapley \citep{garrido2024}, denoted by $\widehat\phi_{i,\mathrm{DU}}^I$. The second is the stratified Monte Carlo Shapley estimator \cite{maleki2013}, denoted by $\widehat\phi_{i,\mathrm{strat}}^I$. Formal definitions of these estimators are given in Appendix~\ref{sec:missing-defs}.

The reason these estimators fit the previous definition is that their cardinality-wise terms are close, in aggregate, to the same first-order signal that defines $\Theta_i^I$. This is made precise in Lemma~\ref{lem:du-close-leading} for DU-Shapley and in Lemma~\ref{lem:strat-close-leading} for stratified Monte Carlo.

\begin{proposition}[DU‑Shapley and stratified Monte Carlo are leading-order estimators]
\label{prop:du_strat}
Let $\widehat{\phi}_{i,\mathrm{DU}}^I$ denote the DU‑Shapley estimator \citep{garrido2024}.  
Let $\widehat{\phi}_{i,\mathrm{strat}}^I$ denote the stratified Monte Carlo Shapley estimator \cite{maleki2013}.  
Then, under the assumptions of Theorem~\ref{thm:leading-term}, both $\widehat{\phi}_{i,\mathrm{DU}}^I$ and $\widehat{\phi}_{i,\mathrm{strat}}^I$ are leading‑order estimators. Consequently,
$\left|\frac{\widehat{\phi}_{i,\mathrm{DU}}^I}{\phi_i^I} - 1\right| \xlongrightarrow{\mathbb{P}} 0,$ and $
\left|\frac{\widehat{\phi}_{i,\mathrm{strat}}^I}{\phi_i^I} - 1\right| \xlongrightarrow{\mathbb{P}} 0.$
\end{proposition}

The proof is given in Appendix~\ref{subsec:proof-prop-du-strat}.
\section{Benchmarking Shapley Estimators using the Leading Term}
\label{sec:benchmarking}
The goal of this section is to illustrate how the leading term can be used as a practical large-scale benchmark for Shapley value approximations. Exact Shapley computation requires averaging over $2^{I-1}$ coalitions for a single player and is therefore unavailable in the large-$I$ regimes where the asymptotic theory is most relevant. In contrast, the leading term can be computed directly whenever the population quantities $\bar n$, $\mu^\star$, and $c_i$ are known. More generally, under the population model, these quantities are estimable from the surrounding data owners through empirical averages, leading to a plug-in estimate of the leading term. A full analysis of the additional error induced by this plug-in estimation is left for future work. In this section, we focus on a controlled synthetic benchmark in which the population quantities are known exactly. This allows us to isolate the behavior of Shapley value approximations relative to the leading term, without confounding it with the statistical error of estimating the benchmark itself.
\paragraph{Setting.} We instantiate the theoretical setup in a finite-dimensional RKHS $\mathcal H = \mathbb R^d$ with $d=30$. Datapoints are represented directly by bounded feature vectors. To obtain a simple and exactly reproducible population model, we use the finite-support construction given in Example \ref{ex:latent-type-population}. There are $T=4$ latent player types, with type probabilities $\pi = (0.15,0.35,0.30,0.20).$ Each type $t\in\{1,\dots,T\}$ is associated with a finite dictionary of $A_t=12$ prototype datapoints $a_{t,1},\dots,a_{t,A_t}\in\mathbb R^d$, where $\|a_{t,\ell}\|_2 = 0.9$. These prototype datapoints are generated once at the beginning of the experiment and kept fixed throughout. Specifically, for each $t$ and each $\ell\in\{1,\dots,A_t\}$, we first draw $g_{t,\ell} \sim \mathcal N(0,I_d)$, independently across $t$ and $\ell$, and then set $a_{t,\ell} = 0.9 \frac{g_{t,\ell}}{\|g_{t,\ell}\|_2}.$ Thus, all prototype datapoints lie on the Euclidean sphere of radius $0.9$, which ensures the bounded-feature condition required by the theory. Conditional on type $t$, datapoints are sampled uniformly from the corresponding prototypes.

For every player $j\neq i$, we generate its dataset as follows. First, we sample a type $\tau_j \sim \mathrm{Categorical}(\pi)$. Then, we sample a dataset size independently of the type, $n_j \sim \mathrm{Unif}\{1,\dots,n_{\max}\},$ with $n_{\max}=50.$ Finally, we sample datapoints $ z_{j,1},\dots,z_{j,n_j}$ i.i.d. from the uniform distribution over the prototypes of type $\tau_j$. This gives a player dataset $D_j = \{z_{j,1},\dots,z_{j,n_j}\}.$ Since the dataset size is independent of the type, the average player size is known exactly: $\bar n = \frac{1+n_{\max}}{2}=25$.

Let $\mu_t = \frac{1}{A_t}\sum_{\ell=1}^{A_t} a_{t,\ell}$ be the mean embedding of type $t$. Because the population distribution is known exactly, the size-weighted population embedding is also known exactly: $\mu^\star = \sum_{t=1}^T \pi_t \mu_t$. This makes the leading term $\Theta_i^I$ an oracle reference in this controlled benchmark. 

\paragraph{Fixed player.} The theory studies the Shapley value of a fixed player $i$ as the number of surrounding players grows. We follow this convention and fix player $i$ once at the beginning of the experiment. To make the first-order signal visible, we choose the fixed player's type as the type whose mean embedding is farthest from the population embedding: $\tau_i \in \arg\max_{t\in\{1,\dots,T\}} \|\mu_t-\mu^\star\|$. We set its dataset size to $n_i=n_{\max}=50$ and sample its datapoints uniformly from the prototype datapoints of type $\tau_i$. Once sampled, $D_i$ is kept fixed throughout all experiments. Its empirical embedding is $\mu_i = \frac{1}{n_i}\sum_{r=1}^{n_i}\varphi(z_{i,r}).$

\paragraph{Utility.}
For a finite multiset $D$, let $\mu(\widehat P_D)=\frac{1}{|D|} \sum_{z\in D}\varphi(z)$ denote its empirical mean embedding, with the convention $\mu(\widehat P_\emptyset)=0$. We define the utility as a smooth functional of
this empirical embedding: $v(D) = F\bigl(\mu(\widehat P_D)\bigr)$ with $F(\cdot) = \tanh\bigl(\beta\langle w,\cdot\rangle\bigr),$ and $\beta=1.5$. The direction $w$ is chosen to align with the fixed player's deviation from the population: $w = \frac{\mu_i-\mu^\star}{\|\mu_i-\mu^\star\|}.$ This choice ensures that the fixed player has a non-negligible first-order signal $c_i = \left\langle \nabla F(\mu^\star), \mu_i-\mu^\star \right\rangle,$
which is the regime in which relative-error convergence is meaningful. The function $F$ is smooth, bounded, and has a bounded derivative and Lipschitz gradient on $\mathbb R^d$, matching the regularity assumptions used in the theoretical analysis.

\paragraph{Expanding the player population.} A key aspect of the experiment is the way in which we vary the number of players $I$. The theory studies the value of a fixed player $i$ as more independent surrounding players are added. To match this regime, we do not draw a new unrelated collection of players for each value of $I$. Instead, we first generate a large collection of players $D_2,\dots,D_{I_{\max}}.$ Then, for any $I\leq I_{\max}$, we define the corresponding game by keeping the fixed player $i$ and using the first $I-1$ players from this collection: $\{i,2,\dots,I\}$. Thus, increasing $I$ amounts to adding new players to an existing population.

\paragraph{Estimators.} We benchmark three estimators against the leading term: $\widehat{\phi}{i,\mathrm{MC}}^I$ with budget $m_I \propto I^2/\log I$, $\widehat{\phi}{i,\mathrm{DU}}^I$, and $\widehat{\phi}{i,\mathrm{strat}}^I$. We do not include $\widehat{\phi}{i,\mathrm{GT}}^I$ in this benchmark because Proposition~\ref{prop:mc_gt_budgets} invokes the group-testing guarantee through a prescribed $(\varepsilon,\delta)$-accuracy for the full Shapley vector, rather than through a simple fixed-player budget scaling. We therefore focus the empirical comparison on estimators whose computational budgets are directly interpretable in the leading-term regime.

\paragraph{Evaluation metric.} For each estimator $\widehat\phi_i^I$, we report its relative error with respect to the leading term: $\mathrm{RelErr}(\widehat\phi_i^I,\Theta_i^I) = \left|\frac{\widehat\phi_i^I}{\Theta_i^I}-1\right|$. We repeat the full experiment $R=30$ times independently. In each repetition, we first draw one large population of players up to $I_{\max}=3500$. We then evaluate all estimators on the same increasing sequence of player populations, with $I \in \{50, 75, 100, 150, 225, 325, 475, 700, 1000, 1300, 1600, 1900, 2100, 2500, 3000, 3500\}$. For each value of $I$, we plot the empirical mean relative error across repetitions. Shaded regions correspond to one standard error of the mean. Because the theory predicts an $O(1/I)$ absolute error to a quantity of scale $\log(I)/I$, the corresponding relative error is expected to decrease at the rate $1/\log I$. We therefore include a rescaled $1/\log I$ reference curve as a visual guide. The resulting benchmark is shown in Figure \ref{fig:benchmark}. All experiments were run on CPU only and completed in less than a minute.

\begin{figure*}[ht]
    \centering
    \subfloat[$\widehat{\phi}_{i,\text{MC}}^I$]{%
        \includegraphics[width=0.325\textwidth]{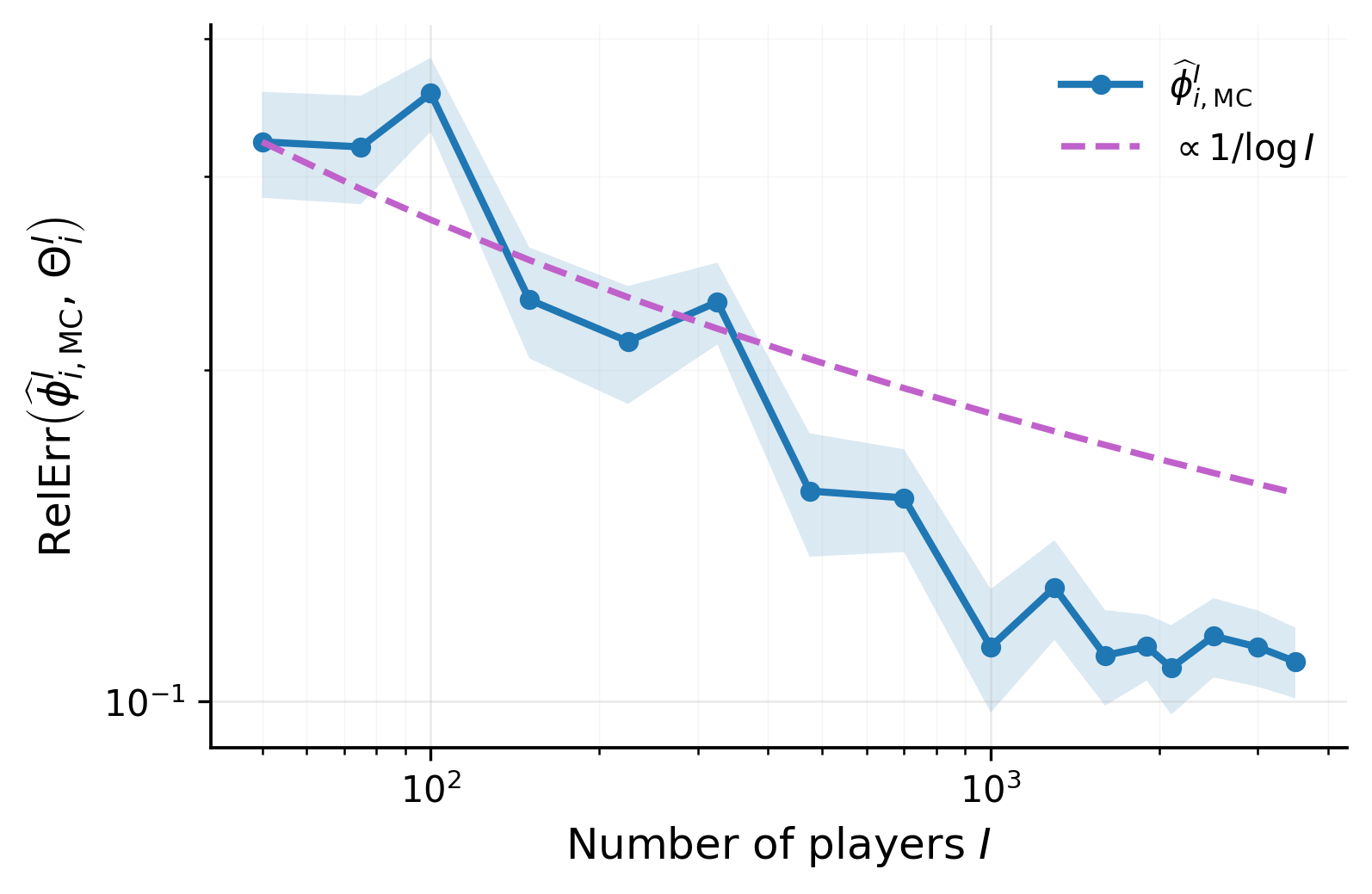}
    }
    \hfill
    \subfloat[$\widehat{\phi}_{i,\text{DU}}^I$]{%
        \includegraphics[width=0.325\textwidth]{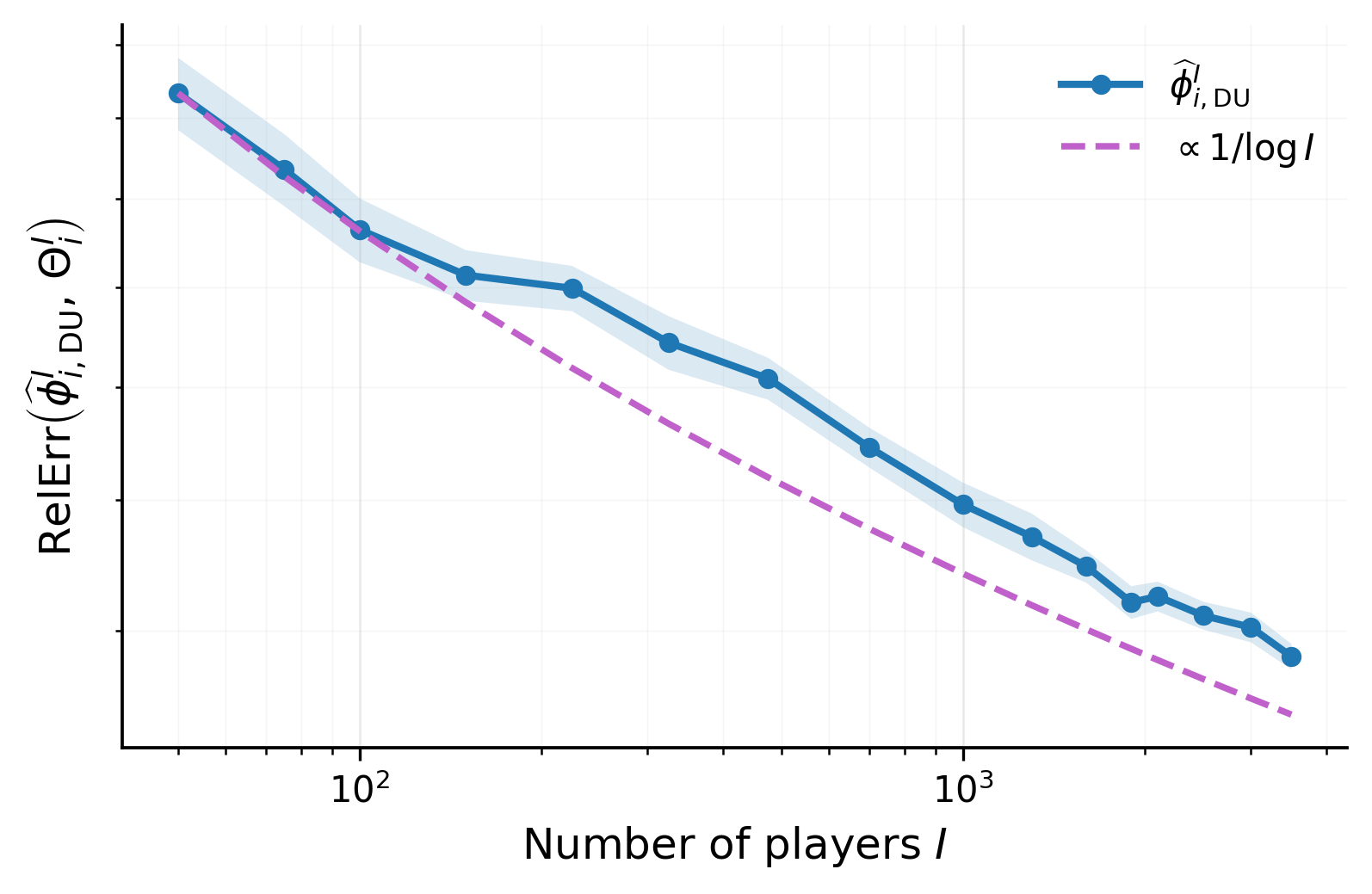}
    }
    \hfill
    \subfloat[$\widehat{\phi}_{i,\text{strat}}^I$]{%
        \includegraphics[width=0.325\textwidth]{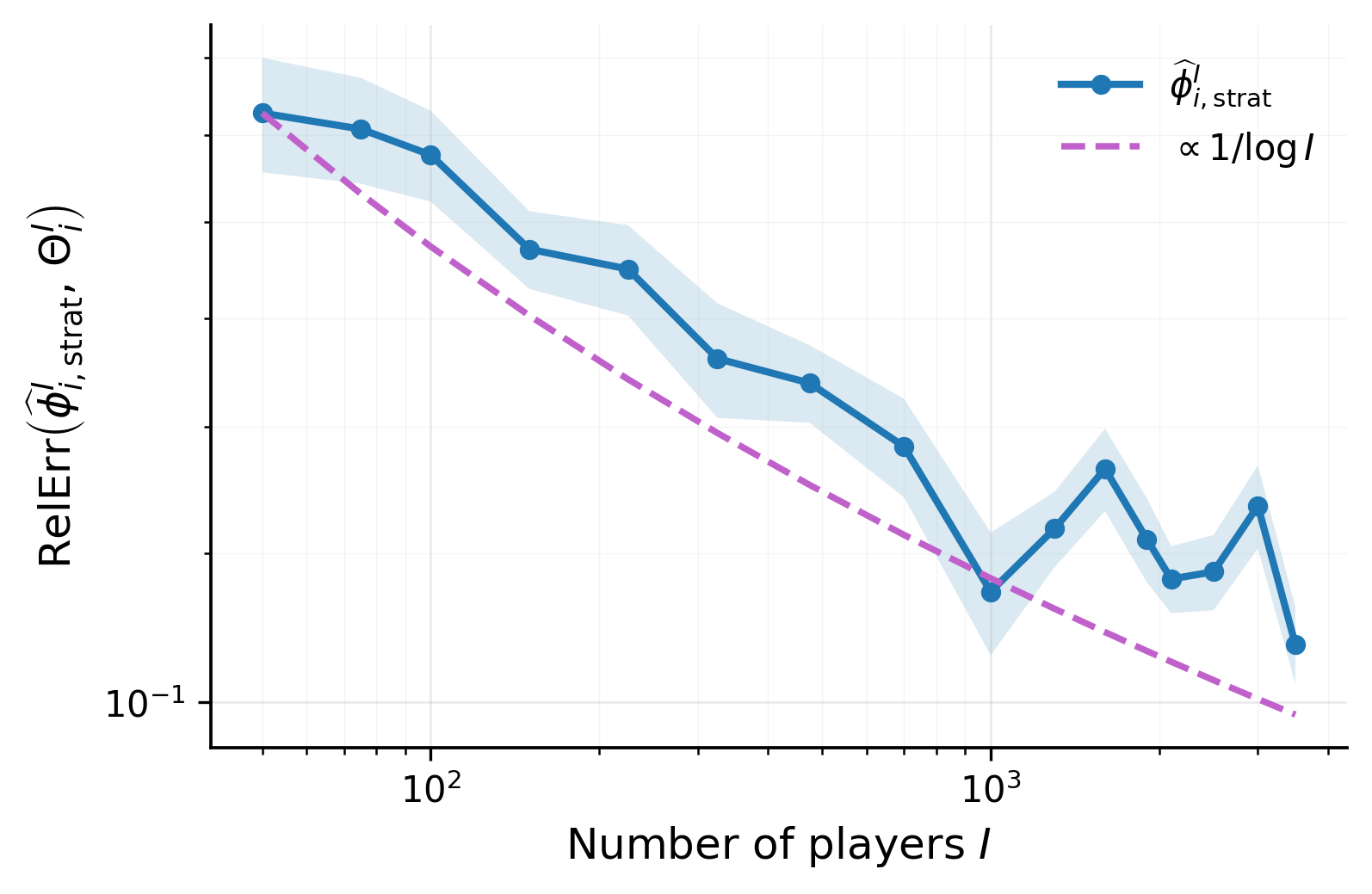}
    }
    \caption{Relative error of each estimator averaged over $30$ independent repetitions with standard errors. The dashed curve is a reference curve included as a visual guide to assess the relative error decay.}
    \label{fig:benchmark}
    \vskip -0.2in
\end{figure*}
\paragraph{Discussion.} All three estimators become closer to the leading term as the number of players increases. This supports the use of $\Theta_i^I$ as a large-scale reference for Shapley value approximations in this controlled setting. For $\widehat{\phi}_{i,\text{DU}}^I$ and $\widehat{\phi}_{i,\text{strat}}^I$, the observed behavior is consistent with the theory: these estimators are leading-order estimators, their relative error is therefore expected to decrease at the scale $1/\log I$. The curve of $\widehat{\phi}_{i,\mathrm{MC}}^I$ decreases faster than the displayed $1/\log I$ reference. This is not a contradiction: the reference is based on a conservative worst-case variance bound for permutation Monte Carlo. In this synthetic model, the variance of a random marginal contribution decreases with $I$, because most predecessor coalitions in a random permutation are large, and the marginal contribution of the fixed player becomes smaller for large coalitions. Therefore, the effective Monte Carlo error is smaller than the worst-case bound, which explains the faster empirical decay.
\section{Conclusion}
\label{sec:conclusion}
This paper develops an asymptotic analysis of the Shapley value for dataset valuation. We study the value of a single fixed data owner as the number of surrounding data sources grows, and show that under a smooth embedding utility model, the exact Shapley value is governed by an explicit and interpretable leading term $\Theta_i^I$. This result makes the asymptotic behavior of the Shapley value interpretable and clarifies what it means to approximate Shapley values in this regime.  
\paragraph{Limitations.} The scope of these results is tied to the assumptions of the analysis. We assume independent surrounding data sources and utilities that are smooth functionals of empirical RKHS mean embeddings. These assumptions facilitate the analysis, but they do not cover all the dataset valuation settings encountered in practice, especially machine learning utilities defined by complex training procedures. We view this work as opening the way to a broader asymptotic theory of Shapley-based dataset valuation, with extensions to richer utilities as natural directions for future work.
\begin{ack}
This work was partially supported by the French National Research Agency (ANR) through grants ANR-20-CE23-0007 and ANR-23-CE23-0002 and through the PEPR IA FOUNDRY project (ANR-23-PEIA-0003).
\end{ack}

\bibliographystyle{plain}
\bibliography{bibliography}

\newpage
\appendix
\paragraph{Outline.}
The appendices are organized as follows.
\begin{itemize}
    \item[--] Appendix~\ref{sec:extended-related-work} provides an extended discussion of related work, including Shapley-based data valuation, efficient Shapley estimation, and alternative valuation methods.
    \item[--] Appendix~\ref{sec:missing-defs} recalls the definitions of the estimators used in Section~\ref{sec:estimators}.
    \item[--] Appendix~\ref{sec:technical-lemmas} collects the technical lemmas used in the proofs.
    \item[--] Appendix~\ref{sec:additional-proofs} contains the proofs of the main theoretical results stated in Sections~\ref{sec:asymptotic-characterization} and~\ref{sec:estimators}.
\end{itemize}

\paragraph{Notation.}
We summarize the main notation used throughout the paper in Table~\ref{tab:notation}.

\begin{table}[ht]
\caption{Notation for the asymptotic analysis of the Shapley value.}
\label{tab:notation}
\begin{center}
\begin{tabular}{c c l}
\toprule
Notation & Type & Description \\
\midrule
$\mathcal Z$ & measurable space & Data domain \\
$\mathcal M(\mathcal Z)$ & set & Set of finite multisets over $\mathcal Z$ \\
$j$ & index & Player / dataset owner index \\
$i$ & index & Fixed player whose value is studied \\
$I$ & $\mathbb N$ & Total number of players in the game \\
$[I]$ & set & Player set $\{1,\dots,I\}$ \\
$D_j$ & $\mathcal M(\mathcal Z)$ & Dataset owned by player $j$ \\
$n_j$ & $\mathbb N$ & Size of dataset $D_j$ \\
$D_S$ & $\mathcal M(\mathcal Z)$ & Pooled dataset of coalition $S$ \\
$n_S$ & $\mathbb N$ & Total number of datapoints in $D_S$, $n_S=\sum_{j\in S}n_j$ \\
$\bar n$ & $\mathbb R_+$ & Mean dataset size of a surrounding player, $\bar n=\mathbb E[n_j]$ \\
$v$ & $\mathcal M(\mathcal Z)\to\mathbb R$ & Utility of a pooled dataset \\
$u_I$ & $2^{[I]}\to\mathbb R$ & Cooperative game utility, $u_I(S)=v(D_S)$ \\
$\phi_i^I$ & $\mathbb R$ & Shapley value of player $i$ in the game with $I$ players \\
$S_k$ & random subset & Uniform subset of $[I]\setminus\{i\}$ of cardinality $k$ \\
\midrule
$\mathcal H$ & Hilbert space & RKHS associated with the kernel \\
$\varphi$ & $\mathcal Z\to\mathcal H$ & Feature map \\
$\kappa$ & $\mathbb R_+$ & Uniform bound on feature embeddings, $\|\varphi(z)\|\le \kappa$ \\
$F$ & $\mathcal H\to\mathbb R$ & Smooth functional defining the utility \\
$G$ & $\mathbb R_+$ & Uniform bound on $\|\nabla F\|$ \\
$M$ & $\mathbb R_+$ & Lipschitz constant of $\nabla F$ \\
$\mu(\widehat P_D)$ & $\mathcal H$ & Empirical mean embedding of a finite dataset $D$ \\
$\mu_j$ & $\mathcal H$ & Empirical embedding of player $j$, $\mu_j=\mu(\widehat P_{D_j})$ \\
$\mu_i$ & $\mathcal H$ & Empirical embedding of the fixed player $i$ \\
$X_j$ & $\mathcal H$ & Feature sum of player $j$, $X_j=\sum_{z\in D_j}\varphi(z)$ \\
$\mu_S$ & $\mathcal H$ & Empirical embedding of coalition $S$, $\mu_S=\mu(\widehat P_{D_S})$ \\
$\mu^\star$ & $\mathcal H$ & Population reference embedding, $\mu^\star=\mathbb E[X_j]/\bar n$ \\
$g^\star$ & $\mathcal H$ & Population gradient, $g^\star=\nabla F(\mu^\star)$ \\
$c_i$ & $\mathbb R$ & First-order signal, $c_i=\langle g^\star,\mu_i-\mu^\star\rangle$ \\
\midrule
$\Delta_i(S)$ & $\mathbb R$ & Marginal contribution $v(D_S\biguplus D_i)-v(D_S)$ \\
$R_S$ & $\mathbb R$ & Taylor remainder in the marginal contribution expansion \\
$\varepsilon_S$ & $\mathbb R$ & Residual after subtracting the population first-order term \\
$H_{I-1}$ & $\mathbb R_+$ & Harmonic number, $H_{I-1}=\sum_{k=1}^{I-1}1/k$ \\
$\Theta_i^I$ & $\mathbb R$ & Leading term, $\Theta_i^I=\frac{n_i c_i}{\bar n}\frac{H_{I-1}}{I}$ \\
\bottomrule
\end{tabular}
\end{center}
\end{table}

\newpage
\section{Extended Related Work}
\label{sec:extended-related-work}

\subsection{Shapley-based data(set) valuation}
\label{subsec:rw-shapley-dataset-valuation}

The use of the Shapley value \cite{shapley1953} for data valuation was introduced to assign a contribution score to individual training examples in supervised learning \citep{ghorbani2019,jia2019b}. It has been widely studied as a data valuation method because it uniquely satisfies four key axioms:
\begin{enumerate}
    \item \emph{Null player.} If a player has zero marginal contribution to every coalition, i.e., $u(S\cup\{i\}) = u(S) \forall S\subseteq \mathcal I\setminus\{i\},$ then its value is zero: $\phi_i(u)=0$.
    \item \emph{Symmetry.} If two players have identical marginal contributions to every coalition, i.e., $u(S\cup\{i\}) = u(S\cup\{j\}) \forall S\subseteq \mathcal I\setminus\{i,j\},$ then they receive the same value: $\phi_i(u)=\phi_j(u)$.
    \item \emph{Efficiency.} The total value is distributed among all players: $\sum_{i\in\mathcal I}\phi_i(u)=u(\mathcal I)-u(\emptyset)$.
    \item \emph{Additivity.} For any two games $u$ and $w$ on the same player set, $\phi_i(u+w)=\phi_i(u)+\phi_i(w),\forall i\in\mathcal I$.
\end{enumerate}
In the game-theoretic formulation of data valuation, each data point is treated as a player in a cooperative game, and the utility of a coalition is typically defined as the performance of a model trained on the corresponding subset. This point-level setting has motivated a large literature on data valuation for noisy-label detection, data selection, data debugging, and data markets \citep{sim2022,jiang2023}, using the Shapley value as the underlying value notion. Dataset valuation extends this perspective from individual datapoints to data owners: each player contributes a finite dataset, and the goal is to value the incremental gain obtained by pooling this dataset with those of other owners \citep{agarwal2019,sim2020,tay2022,garrido2024}. The present paper is in this dataset-level setting.

Several works have proposed alternative Shapley-type or semivalue-based notions for machine learning. Beta-Shapley relaxes the efficiency axiom and uses a family of cardinality weights to reduce variance and improve robustness in data valuation tasks \citep{kwon2021beta}. Data Banzhaf uses the Banzhaf value and shows improved robustness to stochastic perturbations in model performance scores \citep{wang2023data}. More broadly, semivalue-based frameworks study valuation rules that preserve linearity, symmetry, and the dummy-player property while varying the coalition-size weighting scheme \citep{kwon2021beta,lin2022measuring,wang2023data}.
These works study alternative valuation rules or alternative finite-game weightings. In contrast, we keep the Shapley value itself fixed and study how the exact Shapley value of a given data owner evolves as the number of surrounding owners increases.

A recent line of work orthogonal to our paper also studies the robustness and sensitivity of semivalue-based data values, including Shapley values, to modeling choices \cite{tamine2026, diehl2025}.

A related line of work treats the Shapley value as a statistical object. Distributional Shapley replaces the fixed finite dataset with a population-level quantity and studies the expected contribution of samples drawn from an underlying distribution \citep{ghorbani2020,kwon2021efficient}. More recent work studies uncertainty quantification for Data Shapley by relating it to infinite-order U-statistics and deriving confidence intervals for valuation under changes in the data distribution \citep{wu2024uncertainty}. These works are close in spirit because they go beyond a single deterministic finite dataset. However, their object differs from ours: they either define a new population-level valuation or study the sampling distribution of a valuation estimator. Our object is the sequence of exact finite-game Shapley values $(\phi_i^I)_{I\ge1}$ obtained by fixing one data owner and adding surrounding owners sampled from a population.

Finally, classical cooperative game theory has studied values in nonatomic and large games \citep{aumann1974,neyman1981,neyman1988}. These works provide the mathematical foundations for value concepts beyond finite games, but their setting is different from dataset valuation: players are infinitesimal or embedded in a measure-theoretic game, whereas we keep the valued data owner finite and non-negligible. The main difficulty in our setting is therefore not only to define a limiting game, but to understand how the empirical distribution of large data coalitions affects the finite-game Shapley value.

\subsection{Efficient estimation of the Shapley value}
\label{subsec:rw-efficient-estimation}

The exact Shapley value is computationally intractable in general because it averages marginal contributions over all coalitions. A large literature therefore focuses on estimating the Shapley value for a fixed finite game. Permutation Monte Carlo estimators approximate the Shapley value by averaging marginal contributions along random permutations \citep{castro2009,maleki2013,ghorbani2019,jia2019b}. Variance-reduction and sampling-allocation strategies have also been studied to improve the efficiency of permutation-based estimation \citep{maleki2013,mitchell2022}. Group-testing-based estimators provide another approach by querying randomly sampled coalitions and recovering the Shapley vector through pairwise difference estimates \citep{jia2019b}.

Other works exploit structure in particular learning problems. For nearest-neighbor classifiers, exact or nearly exact Shapley values can be computed much more efficiently than by generic coalition enumeration \citep{jia2019a,wang2024a}. More recent methods avoid repeated model retraining by using approximations along a single training trajectory or by learning surrogate models for data values \citep{wang2024b,sun2026,tamine2025}. DU-Shapley is particularly related to our work because it is a dataset-level proxy designed to exploit the structure of dataset valuation rather than treating the utility as an arbitrary cooperative game \citep{garrido2024}. It also considers a regime with many data owners, but its goal is to propose and justify a proxy estimator.

Instead of starting from an estimator and proving a finite-game error bound, we first characterize the asymptotic scale and leading-order structure of the exact Shapley value itself. This then provides a principled way to analyze estimators in the growing-player regime: estimators targeting the exact Shapley value must have an absolute error of $o(\log I/I)$, while leading-order estimators can be studied by their distance to the leading term $\Theta_i^I$.

\subsection{Alternative valuation methods}
\label{subsec:rw-alternative-valuation}

Beyond Shapley-based methods, several approaches assign value to data by approximating influence on training dynamics or model predictions. Influence functions approximate the effect of upweighting or removing a training point through derivatives of the learned parameters \citep{koh2017influence}. TracIn and related methods track the effect of training examples along optimization trajectories by accumulating gradient-based quantities during training \citep{pruthi2020tracin,paul2021data}. Representer-point methods decompose predictions through training examples in the learned representation space \citep{yeh2018representer}. Data models instead learn a predictive model of how training subsets affect model outputs, providing a different way to model data influence through subset-response behavior \citep{ilyas2022datamodels}.

There are also valuation methods designed to reduce or remove the dependence on repeated model training. Robust volume assigns values using geometric diversity criteria and is designed to be validation-free and replication-robust \citep{xu2021validation}. Data-OOB uses out-of-bag estimates from random forests as an efficient proxy for data value \citep{kwon2023dataoob}. LAVA formulates data valuation using distances between training and validation distributions and does not require specifying the downstream learning algorithm in advance \citep{just2023lava}. Other training-free or initialization-based approaches use neural tangent kernel or generalization-bound arguments to estimate the value of data before full training \citep{wu2022davinz}.

These methods address important practical limitations of Shapley-based valuation, especially computational cost and dependence on model retraining. However, they generally define different valuation rules or proxies. Our goal is different: we study the Shapley value itself in a dataset valuation game and characterize its asymptotic behavior. The leading term obtained in our analysis can therefore be used both as an interpretability tool for the Shapley value and as a reference for evaluating Shapley approximations in regimes where exact computation is infeasible.
\newpage
\section{Missing definitions}
\label{sec:missing-defs}

This appendix collects the estimator definitions used in Section~\ref{sec:estimators}. We state them separately from the main text to keep the presentation focused on the asymptotic consequences of Theorem~\ref{thm:leading-term}. Throughout this appendix, $I\ge 2$ is fixed, the player of interest is $i\in[I]$, and
\begin{align*}
\Delta_i(S) := v(D_S\biguplus D_i)-v(D_S), \qquad S\subseteq[I]\setminus\{i\},
\end{align*}
denotes the marginal contribution of player $i$ to a coalition $S$.

\subsection{Permutation Monte Carlo Shapley}
\label{app:def-mc-shapley}

The permutation representation of the Shapley value writes
\begin{align*}
    \phi_i^I
    =
    \mathbb E_{\pi}
    \left[
        \Delta_i(P_i^\pi)
    \right],
\end{align*}
where $\pi$ is a uniformly random permutation of $[I]$, and $P_i^\pi\subseteq[I]\setminus\{i\}$ is the set of players appearing before $i$ in the order $\pi$. The permutation Monte Carlo Shapley estimator approximates this expectation by averaging over independent random permutations.

Let $\pi_1,\dots,\pi_{m_I}$ be independent uniformly random permutations of $[I]$. For each $b\in\{1,\dots,m_I\}$, let $P_i^{\pi_b}$ be the set of
predecessors of $i$ in $\pi_b$. The permutation Monte Carlo estimator \citep{maleki2013, jia2019b, ghorbani2019} is
\begin{align*}
    \widehat\phi_{i,\mathrm{MC}}^I
    :=
    \frac{1}{m_I}
    \sum_{b=1}^{m_I}
    \Delta_i(P_i^{\pi_b})
    =
    \frac{1}{m_I}
    \sum_{b=1}^{m_I}
    \left[
        v(D_{P_i^{\pi_b}}\biguplus D_i)
        -
        v(D_{P_i^{\pi_b}})
    \right].
\end{align*}
Conditionally on the realized $I$-player game, this estimator is unbiased for $\phi_i^I$.

\subsection{Group-testing Shapley}
\label{app:def-group-testing-shapley}

We recall the group-testing estimator of \cite{jia2019b}. Since adding a constant to the utility does not change marginal contributions or Shapley values, define the normalized game
\begin{align*}
    U_I(S):=v(D_S)-v(\emptyset),
    \qquad
    S\subseteq[I].
\end{align*}
Thus $U_I(\emptyset)=0$, as in the standard group-testing formulation.

Let
\begin{align*}
    Z_I
    :=
    2\sum_{k=1}^{I-1}\frac{1}{k},
\end{align*}
and define a probability distribution on coalition sizes
$\{1,\dots,I-1\}$ by
\begin{align*}
    q_I(k)
    :=
    \frac{1}{Z_I}
    \left(
        \frac{1}{k}
        +
        \frac{1}{I-k}
    \right).
\end{align*}
For each query $t\in\{1,\dots,T_I\}$, draw $K_t\sim q_I$, then draw a coalition $A_t\subseteq[I]$ uniformly among subsets of cardinality $K_t$, and evaluate $U_I(A_t)$. For each player $j\in[I]$, define
\begin{align*}
    \beta_{t,j}:=\mathbf{1}\{j\in A_t\}.
\end{align*}
For each pair $p,q\in[I]$, estimate the Shapley difference
$\phi_p^I-\phi_q^I$ by
\begin{align*}
    \widehat\Delta_{p,q}
    :=
    \frac{Z_I}{T_I}
    \sum_{t=1}^{T_I}
    U_I(A_t)
    \left(
        \beta_{t,p}-\beta_{t,q}
    \right).
\end{align*}
The estimated Shapley vector
\begin{align*}
    \widehat s^I
    =
    (\widehat s_1^I,\dots,\widehat s_I^I)
\end{align*}
is then obtained by solving the feasibility problem
\begin{align*}
    \sum_{j=1}^I \widehat s_j^I
    =
    U_I([I]),
\end{align*}
together with the pairwise constraints
\begin{align*}
    \left|
        (\widehat s_p^I-\widehat s_q^I)
        -
        \widehat\Delta_{p,q}
    \right|
    \le
    \frac{\varepsilon_I}{2\sqrt I},
    \qquad
    \forall p,q\in[I].
\end{align*}
The group-testing Shapley estimate of the fixed player $i$ is
\begin{align*}
    \widehat\phi_{i,\mathrm{GT}}^I
    :=
    \widehat s_i^I.
\end{align*}

\subsection{DU-Shapley}
\label{app:def-du-shapley}

DU-Shapley \citep{garrido2024} replaces coalitions of $k$ players by pseudo-coalitions containing approximately the same number of datapoints as a typical coalition of $k$ surrounding players.

Let
\begin{align*}
    D_{-i}^I
    :=
    D_{[I]\setminus\{i\}}
\end{align*}
be the pooled multiset of all datapoints from players other than $i$, and let
\begin{align*}
    \widehat n_{-i}^I
    :=
    \frac{1}{I-1}
    \sum_{j\in[I]\setminus\{i\}} n_j
\end{align*}
be the empirical average dataset size among the surrounding players. For each $k\in\{0,\dots,I-1\}$, define the target pseudo-coalition size
\begin{align*}
    m_{k,I}
    :=
    \left\lfloor
        k\widehat n_{-i}^I
    \right\rfloor.
\end{align*}
Let $\widetilde D_{k,I}$ be a multiset of $m_{k,I}$ datapoints sampled uniformly without replacement from $D_{-i}^I$. Define
\begin{align*}
    \widetilde\Delta_{k,I}
    :=
    v(\widetilde D_{k,I}\biguplus D_i)
    -
    v(\widetilde D_{k,I}).
\end{align*}
The DU-Shapley estimator is
\begin{align*}
    \widehat\phi_{i,\mathrm{DU}}^I
    :=
    \frac{1}{I}
    \sum_{k=0}^{I-1}
    \widetilde\Delta_{k,I}.
\end{align*}
This estimator preserves the Shapley averaging over coalition sizes but replaces player-level coalitions with data point-level pseudo-coalitions sampled from the surrounding population.

\subsection{Stratified Monte Carlo Shapley}
\label{app:def-stratified-mc-shapley}

The stratified Monte Carlo Shapley estimator \cite{maleki2013} samples coalitions separately at each cardinality. For every $k\in\{0,\dots,I-1\}$, let
$m_{k,I}\ge 1$ be the number of samples allocated to stratum $k$. For each $b\in\{1,\dots,m_{k,I}\}$, draw
\begin{align*}
    S_{k,b}
    \subseteq
    [I]\setminus\{i\}
\end{align*}
uniformly among subsets of cardinality $k$. The stratified estimator is
\begin{align*}
    \widehat\phi_{i,\mathrm{strat}}^I
    :=
    \frac{1}{I}
    \sum_{k=0}^{I-1}
    \frac{1}{m_{k,I}}
    \sum_{b=1}^{m_{k,I}}
    \Delta_i(S_{k,b}).
\end{align*}
Equivalently,
\begin{align*}
    \widehat\phi_{i,\mathrm{strat}}^I
    =
    \frac{1}{I}
    \sum_{k=0}^{I-1}
    \frac{1}{m_{k,I}}
    \sum_{b=1}^{m_{k,I}}
    \left[
        v(D_{S_{k,b}}\biguplus D_i)
        -
        v(D_{S_{k,b}})
    \right].
\end{align*}
Conditionally on the realized $I$-player game, this estimator is unbiased for $\phi_i^I$, because each stratum estimates the corresponding cardinality-wise expectation in the Shapley formula.
\newpage
\section{Technical lemmas}
\label{sec:technical-lemmas}
\begin{lemma}[Embedding update when adding a player]
\label{lem:update}
Let $S$ be a coalition that does not contain $i$, with pooled data point count $n_S$.
Let
\begin{align*}
\mu_S := \mu(\widehat P_{D_S})
\quad\text{and}\quad
\mu_{S\cup\{i\}} := \mu(\widehat P_{D_{S\cup\{i\}}}).
\end{align*}
Then
\begin{align}
\mu_{S\cup\{i\}} = \mu_S + \alpha_S(\mu_i-\mu_S),
\qquad
\alpha_S := \frac{n_i}{n_S+n_i}.\\
\label{eq:update}
\end{align}
\end{lemma}
\begin{proof}
By definition,
\begin{align*}
\mu_S=\frac{1}{n_S}\sum_{j\in S}\sum_{r=1}^{n_j}\phi(z_{j,r}),
\qquad
\mu_{S\cup\{i\}}=\frac{1}{n_S+n_i}\left(\sum_{j\in S}\sum_{r=1}^{n_j}\phi(z_{j,r})+\sum_{r=1}^{n_i}\phi(z_{i,r})\right).
\end{align*}
Rewrite the numerator using $n_S \mu_S$ and $n_i\mu_i$:
\begin{align*}
\mu_{S\cup\{i\}} = \frac{1}{n_S+n_i}\bigl(n_S \mu_S + n_i\mu_i\bigr)
= \frac{n_S}{n_S+n_i}\mu_S+\frac{n_i}{n_S+n_i}\mu_i.
\end{align*}
Subtract and add $\mu_S$:
\begin{align*}
\mu_{S\cup\{i\}}-\mu_S
=
\left(\frac{n_S}{n_S+n_i}-1\right)\mu_S+\frac{n_i}{n_S+n_i}\mu_i
=
-\frac{n_i}{n_S+n_i}\mu_S+\frac{n_i}{n_S+n_i}\mu_i
=
\alpha_S(\mu_i-\mu_S).
\end{align*}
This is exactly \eqref{eq:update}.
\end{proof}

\begin{lemma}[First-order Taylor remainder bound in a Hilbert space with Lipschitz gradient]
\label{lem:taylor}
Under \cref{ass:smooth-utility}, for all $x,h\in\cH$,
\begin{align*}
\bigl|F(x+h)-F(x)-\ip{\nabla F(x)}{h}\bigr|
\le \frac{M}{2}\|h\|^2.
\end{align*}
\end{lemma}

\begin{proof}
Define $\psi(t):=F(x+th)$ for $t\in[0,1]$.
By Fréchet differentiability, $\psi$ is differentiable with
\begin{align*}
\psi'(t)=\ip{\nabla F(x+th)}{h}.
\end{align*}
By the fundamental theorem of calculus,
\begin{align*}
F(x+h)-F(x) - \ip{\nabla F(x)}{h}
&=
\int_0^1 \bigl(\psi'(t)-\psi'(0)\bigr)\,dt
=
\int_0^1 \ip{\nabla F(x+th)-\nabla F(x)}{h}\,dt.
\end{align*}
Take absolute values and apply Cauchy-Schwarz:
\begin{align*}
\bigl|F(x+h)-F(x)-\ip{\nabla F(x)}{h}\bigr|
\le \int_0^1 \|\nabla F(x+th)-\nabla F(x)\|\ \|h\|\,dt.
\end{align*}
Using the $M$-Lipschitz property of $\nabla F$,
\begin{align*}
\|\nabla F(x+th)-\nabla F(x)\|\le M\|th\| = Mt\|h\|.
\end{align*}
Therefore
\begin{align*}
\bigl|F(x+h)-F(x)-\ip{\nabla F(x)}{h}\bigr|
\le \int_0^1 Mt\|h\|^2\,dt
= \frac{M}{2}\|h\|^2.
\end{align*}
\end{proof}
\noindent For a set $S$ of size $k$, define $\mu_S=\mu(\widehat P_{D_S})$ as above.
\begin{lemma}[First-order expansion of a marginal contribution]
\label{lem:first-order-marginal-expansion}
Under Assumptions~\ref{ass:bounded-feature} and~\ref{ass:smooth-utility}, for
any nonempty coalition $S\subseteq [I]\setminus\{i\}$,
\begin{align*}
    \Delta_i(S)
    =
    \alpha_S
    \left\langle
        \nabla F(\mu_S),
        \mu_i-\mu_S
    \right\rangle
    +
    R_S,
    \qquad
    \alpha_S:=\frac{n_i}{n_S+n_i},
\end{align*}
where
\begin{align*}
    |R_S|
    \le
    2M\kappa^2\alpha_S^2.
\end{align*}
\end{lemma}

\begin{proof}
By Lemma~\ref{lem:update},
\begin{align*}
    \mu_{S\cup\{i\}}=\mu_S+\alpha_S(\mu_i-\mu_S).
\end{align*}
Since $v(D)=F(\mu(\widehat P_D))$,
\begin{align*}
    \Delta_i(S)
    =
    F\!\left(\mu_S+\alpha_S(\mu_i-\mu_S)\right)-F(\mu_S).
\end{align*}
Apply Lemma~\ref{lem:taylor} with
\begin{align*}
    x=\mu_S,
    \qquad
    h=\alpha_S(\mu_i-\mu_S).
\end{align*}
This gives the stated expansion with
\begin{align*}
    |R_S|
    \le
    \frac{M}{2}\alpha_S^2\|\mu_i-\mu_S\|^2.
\end{align*}
Because $\|\mu_i\|\le\kappa$ and $\|\mu_S\|\le\kappa$,
\begin{align*}
    \|\mu_i-\mu_S\|\le 2\kappa.
\end{align*}
Hence
\begin{align*}
    |R_S|
    \le
    2M\kappa^2\alpha_S^2.
\end{align*}
\end{proof}
\begin{lemma}[Concentration of coalition embeddings]
\label{lem:coalition-embedding-concentration}
Under Assumptions~\ref{ass:population} and~\ref{ass:bounded-feature}, for any
fixed $k\ge 1$ and any subset $S\subseteq \mathbb N\setminus\{i\}$ of size
$k$ generated from the surrounding-player population,
\begin{align*}
    \mathbb E\|\mu_S-\mu^\star\|
    \le
    \frac{2n_{\max}\kappa}{\sqrt{k}}.
\end{align*}
\end{lemma}

\begin{proof}
Let $S$ be a coalition of $k$ surrounding players. Define
\begin{align*}
    Y_j:=X_j-n_j\mu^\star.
\end{align*}
Then
\begin{align*}
    \mu_S-\mu^\star
    =
    \frac{1}{n_S}\sum_{j\in S}Y_j.
\end{align*}
Since $n_j\ge 1$, we have $n_S\ge k$, and therefore
\begin{align*}
    \|\mu_S-\mu^\star\|
    \le
    \frac{1}{k}
    \left\|\sum_{j\in S}Y_j\right\|.
\end{align*}
Taking expectations and applying Cauchy--Schwarz,
\begin{align*}
    \mathbb E\|\mu_S-\mu^\star\|
    \le
    \frac{1}{k}
    \left(
        \mathbb E
        \left\|\sum_{j\in S}Y_j\right\|^2
    \right)^{1/2}.
\end{align*}
By definition,
\begin{align*}
    \mu^\star=\frac{\mathbb E[X_j]}{\bar n},
    \qquad
    \bar n=\mathbb E[n_j],
\end{align*}
so
\begin{align*}
    \mathbb E[Y_j]
    =
    \mathbb E[X_j]-\mathbb E[n_j]\mu^\star
    =
    0.
\end{align*}
The variables $Y_j$ are i.i.d. and centered. Hence the cross terms vanish and
\begin{align*}
    \mathbb E
    \left\|
        \sum_{j\in S}Y_j
    \right\|^2
    =
    k\,\mathbb E\|Y_j\|^2.
\end{align*}
Moreover,
\begin{align*}
    \|X_j\|\le n_{\max}\kappa
    \qquad\text{and}\qquad
    \|\mu^\star\|
    \le
    \frac{\mathbb E\|X_j\|}{\bar n}
    \le
    \frac{\mathbb E[n_j]\kappa}{\bar n}
    =
    \kappa.
\end{align*}
Therefore,
\begin{align*}
    \|Y_j\|
    \le
    \|X_j\|+n_j\|\mu^\star\|
    \le
    2n_{\max}\kappa.
\end{align*}
Combining the previous displays gives
\begin{align*}
    \mathbb E\|\mu_S-\mu^\star\|
    \le
    \frac{1}{k}
    \sqrt{k(2n_{\max}\kappa)^2}
    =
    \frac{2n_{\max}\kappa}{\sqrt{k}}.
\end{align*}
\end{proof}
\begin{lemma}[First-order approximation of the Shapley marginal contribution and its error bound]
\label{lem:marginal_bound}
Under Assumptions \ref{ass:bounded-feature} and \ref{ass:smooth-utility}, for any coalition $S$ not containing $i$, 
\begin{align}
|\Delta_i(S)-\alpha_S c_i|
\le \alpha_S(G+2M\kappa)\|\mu_S-\mu_\star\| + 2M\kappa^2\,\alpha_S^2.
\label{eq:marginal_bound}
\end{align}
where $\Delta_i(S):=v(D_{S}\biguplus D_i)-v(D_{S})$.
\end{lemma}

\begin{proof}
Let $\alpha:=\alpha_S$ and $x:=\mu_S$.
By \cref{lem:taylor} with $h=\alpha(\mu_i-x)$,
\begin{align*}
F(x+\alpha(\mu_i-x))-F(x) = \alpha\ip{\nabla F(x)}{\mu_i-x}+R
\quad\text{with}\quad |R|\le \frac{M}{2}\alpha^2\|\mu_i-x\|^2.
\end{align*}
Since $\|\mu_i\|\le \kappa$ and $\|x\|=\|\mu(\widehat P_{D_S})\|\le \kappa$,
\begin{align*}
\|\mu_i-x\|\le \|\mu_i\|+\|x\|\le 2\kappa,
\end{align*}
hence $|R|\le \frac{M}{2}\alpha^2(2\kappa)^2=2M\kappa^2\alpha^2$.

\noindent Now compare $\alpha\ip{\nabla F(x)}{\mu_i-x}$ with $\alpha c_i=\alpha\ip{g_\star}{\mu_i-\mu_\star}$.
First, rewrite:
\begin{align*}
\alpha\ip{\nabla F(x)}{\mu_i-x}-\alpha\ip{g_\star}{\mu_i-\mu_\star}
=\alpha\ip{\nabla F(x)-g_\star}{\mu_i-x} - \alpha\ip{g_\star}{x-\mu_\star}.
\end{align*}
Bound the two terms:
\begin{align*}
\bigl|\alpha\ip{\nabla F(x)-g_\star}{\mu_i-x}\bigr|
\le \alpha\|\nabla F(x)-g_\star\|\ \|\mu_i-x\|
\le \alpha\cdot M\|x-\mu_\star\|\cdot 2\kappa
=2M\kappa\,\alpha\|x-\mu_\star\|,
\end{align*}
where we used the Lipschitz continuity of $\nabla F$.
Also
\begin{align*}
\bigl|\alpha\ip{g_\star}{x-\mu_\star}\bigr|
\le \alpha\|g_\star\|\ \|x-\mu_\star\|
\le \alpha G\|x-\mu_\star\|.
\end{align*}
Combine with the bound on $R$ to get \eqref{eq:marginal_bound}.
\end{proof}

\begin{lemma}[A deterministic bound on $\alpha_S$]
\label{lem:alpha_bound}
For any coalition $S$ of size $k\ge 1$, the coefficient $\alpha_S = \frac{n_i}{n_S + n_i}$ satisfies
\begin{align*}
\alpha_S \le \min\left\{1,\frac{n_i}{k}\right\} \qquad\text{and}\qquad \alpha_S^2 \le \min\left\{1,\frac{n_i^2}{k^2}\right\}.
\end{align*}
\end{lemma}

\begin{proof}
Since each player contributes at least one datapoint, any coalition of size $k$ has total $n_S \ge k$ datapoints. Consequently,
\begin{align*}
\alpha_S = \frac{n_i}{n_S + n_i} \le \frac{n_i}{k + n_i} \le \frac{n_i}{k},
\end{align*}
where the last inequality uses $k + n_i \ge k$. Moreover, from the definition $\alpha_S = n_i/(n_S+n_i)$, we have $\alpha_S \le 1$ because $n_S+n_i \ge n_i$. Combining these two bounds yields $\alpha_S \le \min\{1, n_i/k\}$. For the squared term, squaring the bound gives $\alpha_S^2 \le \min\{1, n_i^2/k^2\}$.
\end{proof}

\begin{lemma}[Corresponds to Eq.(2) in \cite{schneider2016} : without-replacement RKHS embedding concentration]
\label{thm:schneider}
Let $z_1,\dots,z_N$ be a finite population and let $P_N:=\frac{1}{N}\sum_{j=1}^N\delta_{z_j}$ be its empirical measure.
Let $P_n$ be the empirical measure of a subset of size $n$ sampled uniformly \emph{without replacement} from $\{z_1,\dots,z_N\}$.
Assume $\|\phi(z)\|\le d$ for all population points.
Then for all $\varepsilon>0$,
\begin{align}
\Prob\Bigl(\|\mu(P_n)-\mu(P_N)\|\ge \varepsilon\Bigr)
\le 2\exp\!\left(-\frac{n\varepsilon^2}{8d^2}\right).
\label{eq:schneider}
\end{align}
\end{lemma}

\begin{corollary}[Expected deviation of without‑replacement subsampling]
\label{cor:schneider_exp}
Under the assumptions of Lemma \ref{thm:schneider},
\begin{align*}
\E\|\mu(P_n)-\mu(P_N)\|
\le d\sqrt{\frac{8\pi}{n}}.
\end{align*}
\end{corollary}

\begin{proof}
Let $X:=\|\mu(P_n)-\mu(P_N)\|\ge 0$.
Use the identity $\E[X]=\int_0^\infty \Prob(X\ge t)\,dt$ and apply \eqref{eq:schneider}:
\begin{align*}
\E[X]\le \int_0^\infty 2\exp\!\left(-\frac{n t^2}{8d^2}\right)\,dt.
\end{align*}
Let $a:=\frac{n}{8d^2}>0$. Then
\begin{align*}
\int_0^\infty 2e^{-a t^2}\,dt
=2\cdot \frac{1}{2}\sqrt{\frac{\pi}{a}}
=\sqrt{\frac{\pi}{a}}
=d\sqrt{\frac{8\pi}{n}}.
\end{align*}
\end{proof}
\begin{lemma}[Harmonic asymptotics]
\label{lem:harmonic-asymptotics}
Let
\begin{align*}
    H_n:=\sum_{k=1}^n\frac{1}{k}.
\end{align*}
Then
\begin{align*}
    H_n=\log n+O(1).
\end{align*}
Consequently,
\begin{align*}
    \frac{H_{I-1}}{\log I}\longrightarrow 1.
\end{align*}
\end{lemma}

\begin{proof}
For every $n\ge 1$, the standard integral comparison gives
\begin{align*}
    \int_1^{n+1}\frac{1}{x}\,dx
    \le
    H_n
    \le
    1+\int_1^n\frac{1}{x}\,dx.
\end{align*}
Hence
\begin{align*}
    \log(n+1)
    \le
    H_n
    \le
    1+\log n.
\end{align*}
This implies $H_n=\log n+O(1)$. Taking $n=I-1$, we get
\begin{align*}
    H_{I-1}
    =
    \log(I-1)+O(1)
    =
    \log I+O(1),
\end{align*}
and therefore
\begin{align*}
    \frac{H_{I-1}}{\log I}
    =
    1+O\!\left(\frac{1}{\log I}\right)
    \longrightarrow 1.
\end{align*}
\end{proof}
\begin{lemma}[Lower bound on the leading term]
\label{lem:leading-lower-bound}
For every integer $I\ge 2$,
\begin{align*}
    H_{I-1}
    :=
    \sum_{k=1}^{I-1}\frac{1}{k}
    \ge
    \log I .
\end{align*}
Consequently, if
\begin{align*}
    \Theta_i^I
    =
    \frac{n_i c_i}{\bar n}\frac{H_{I-1}}{I},
\end{align*}
then
\begin{align*}
    |\Theta_i^I|
    \ge
    \frac{n_i |c_i|}{\bar n}
    \frac{\log I}{I}.
\end{align*}
\end{lemma}

\begin{proof}
For every integer $n\ge 1$, since $x\mapsto 1/x$ is non-increasing on
$[1,\infty)$,
\begin{align*}
    \int_1^{n+1}\frac{1}{x}\,dx
    \le
    \sum_{k=1}^{n}\frac{1}{k}.
\end{align*}
Taking $n=I-1$, we obtain
\begin{align*}
    \log I
    =
    \int_1^{I}\frac{1}{x}\,dx
    \le
    \sum_{k=1}^{I-1}\frac{1}{k}
    =
    H_{I-1}.
\end{align*}
The bound on $|\Theta_i^I|$ follows immediately from the definition of
$\Theta_i^I$.
\end{proof}

\newpage
\section{Additional Proofs}
\label{sec:additional-proofs}
This appendix contains the missing proofs for the results stated in the main paper. It is organized as follows.
\begin{itemize}
    \item[--] Appendix~\ref{subsec:proof-thm-leading-term} introduces the proof of Theorem~\ref{thm:leading-term}.
    \item[--] Appendix~\ref{subsec:proof-cor-shapley-lower-bound} introduces the proof of  Corollary~\ref{cor:shapley-lower-bound}.
    \item[--] Appendix~\ref{subsec:proof-cor-absolute-to-relative} introduces the proof of  Corollary~\ref{cor:absolute-to-relative}.
    \item[--] Appendix~\ref{subsec:proof-prop-mc-gt} introduces the proof of  Proposition~\ref{prop:mc_gt_budgets}. 
    \item[--] Appendix~\ref{subsec:proof-thm-leading-order-estimator} introduces the proof of Theorem ~\ref{thm:leading_order_consistent}.
    \item[--] Appendix~\ref{subsec:proof-prop-du-strat} introduces the proof of  Proposition~\ref{prop:du_strat}.
\end{itemize}
\newpage

\subsection{Proof of Theorem \ref{thm:leading-term}}
\label{subsec:proof-thm-leading-term}
\begin{theorem}[Restate of Theorem \ref{thm:leading-term}]
Suppose Assumptions~\ref{ass:population}--\ref{ass:smooth-utility} hold. Then,
for the fixed player $i$,
\begin{align*}
    \mathbb E
    \left[
        \left|
            \phi_i^I-\Theta_i^I
        \right|
    \right]
    =
    O\!\left(\frac{1}{I}\right),
\end{align*}
where the expectation is over the surrounding players
$(D_j)_{j\neq i}$.
\end{theorem}
\begin{proof}
Recall the definition of the Shapley value from Equation \eqref{eq:shapley-I-cardinality}:
\begin{align*}
\phi_i^I = \frac{1}{I}\sum_{k=0}^{I-1} \E_{S_k}[\Delta_i(S_k)],
\end{align*}
where $\Delta_i(S_k)$ denotes the marginal contribution $v(D_{S_k}\biguplus D_i)-v(D_{S_k})$ and $S_k$ is uniformly distributed over subsets of $[I]\setminus\{i\}$ of size $k$.
The term $k=0$ corresponds to the empty coalition. By the convention stated in Section \ref{sec:setup}, $\Delta_i(\varnothing)=v(D_i)-v(0)$. Since $\|\mu_i\|\le \kappa$ and $F$ is $G$-Lipschitz (from $\|\nabla F\|\le G$), we have
\begin{align*}
|\Delta_i(\varnothing)| = |F(\mu_i)-F(0)| \le G\|\mu_i-0\| \le G\kappa.
\end{align*}
Thus $|\Delta_i(\varnothing)|$ is bounded by a constant independent of $I$, and its contribution to $\mathbb{E}|\phi_i^I-\Theta_i^I|$ is $O(1/I)$ after multiplication by $1/I$.
We therefore focus on $k\ge 1$ and write
\begin{align*}
\phi_i^I - \Theta_i^I = \frac{1}{I}\sum_{k=1}^{I-1}\E_{S_k}[\Delta_i(S_k)] - \Theta_i^I + \frac{C_0}{I},
\end{align*}
where $C_0$ absorbs the $k=0$ term.

\paragraph{1. Decomposing the error.}
For each coalition $S_k$, we have:
\begin{align*}
\Delta_i(S_k) = \alpha_{S_k}c_i + \bigl(\Delta_i(S_k) - \alpha_{S_k}c_i\bigr).
\end{align*}
Taking expectation over $S_k$ and averaging over $k$, we obtain
\begin{align*}
\phi_i^I - \Theta_i^I = \underbrace{\frac{1}{I}\sum_{k=1}^{I-1}\E_{S_k}\bigl[\Delta_i(S_k)-\alpha_{S_k}c_i\bigr]}_{=:A_I}
+ \underbrace{\left(\frac{1}{I}\sum_{k=1}^{I-1}\E_{S_k}[\alpha_{S_k}]c_i - \Theta_i^I\right)}_{=:m_I}
+ \frac{C_0}{I},
\end{align*}
By the triangle inequality,
\begin{align*}
\E|\phi_i^I-\Theta_i^I| \le \E|A_I| + \E|m_I| + \frac{C_0}{I}.
\end{align*}
We now bound $\E|A_I|$ and $\E|m_I|$ separately.

\paragraph{2. Upper bound on $\E|A_I|$}
We have
\begin{align*}
\E|A_I| &= \E\left[ \left| \frac{1}{I}\sum_{k=1}^{I-1} \E_{S_k}[\Delta_i(S_k) - \alpha_{S_k}c_i] \right| \right] \\
&\le \frac{1}{I}\sum_{k=1}^{I-1} \E\left[ \E_{S_k}\bigl[ |\Delta_i(S_k) - \alpha_{S_k}c_i| \bigr] \right] \quad \text{(by Jensen and triangle inequality)}\\
&\le \frac{1}{I}\sum_{k=1}^{I-1} \E\left[ (G+2M\kappa)\E_{S_k}[\alpha_{S_k}|\mu_{S_k}-\mu_\star|] + 2M\kappa^2\E_{S_k}[\alpha_{S_k}^2] \right] \text{(by 
Lemma~\ref{lem:marginal_bound})} 
\end{align*}
Now, by the law of total expectation, $\E[\E_{S_k}[\alpha_{S_k}\|\mu_{S_k}-\mu_\star\|]] = \E[\alpha_{S_k}\|\mu_{S_k}-\mu_\star\|]$, where the outer expectation is over the random players (which determine the distribution of $S_k$). The same holds for $\alpha_{S_k}^2$. Hence,
\begin{align*}
\E|A_I| \le \frac{1}{I}\sum_{k=1}^{I-1}\Bigl[(G+2M\kappa)\E\bigl[\alpha_{S_k}\|\mu_{S_k}-\mu_\star\|\bigr] + 2M\kappa^2\E[\alpha_{S_k}^2]\Bigr].
\end{align*}

\subparagraph{2.1 Bound on $\E[\alpha_{S_k}\|\mu_{S_k}-\mu_\star\|]$.}
By Lemma~\ref{lem:alpha_bound}, $\alpha_{S_k} \le \min\{1,n_i/k\}$.
By Lemma~\ref{lem:coalition-embedding-concentration}, $\E\|\mu_{S_k}-\mu_\star\| \le 2n_{\max}\kappa/\sqrt{k}$.
Hence,
\begin{align*}
\E\bigl[\alpha_{S_k}\|\mu_{S_k}-\mu_\star\|\bigr] \le \min\Bigl\{1,\frac{n_i}{k}\Bigr\}\cdot \frac{2n_{\max}\kappa}{\sqrt{k}}.
\end{align*}
Summing over $k$,
\begin{align*}
\sum_{k=1}^{I-1}\E\bigl[\alpha_{S_k}\|\mu_{S_k}-\mu_\star\|\bigr]
\le 2n_{\max}\kappa\left(\sum_{k=1}^{\lfloor n_i\rfloor}\frac{1}{\sqrt{k}} + \sum_{k=\lfloor n_i\rfloor+1}^{I-1}\frac{n_i}{k^{3/2}}\right).
\end{align*}
The first sum is bounded by $2\sqrt{n_i}$ (since $\sum_{k=1}^{n_i}k^{-1/2}\le 2\sqrt{n_i}$).
The second sum is bounded by $n_i\int_{n_i}^{\infty}x^{-3/2}dx = 2\sqrt{n_i}$.
Thus the total sum is at most $8n_{\max}\kappa\sqrt{n_i}$, a finite constant independent of $I$.

\subparagraph{2.2 Bound on $\E[\alpha_{S_k}^2]$.}
Again by Lemma~\ref{lem:alpha_bound}, $\E[\alpha_{S_k}^2] \le \min\{1,n_i^2/k^2\}$.
Summing over $k$,
\begin{align*}
\sum_{k=1}^{I-1}\E[\alpha_{S_k}^2] \le \sum_{k=1}^{\lfloor n_i\rfloor}1 + \sum_{k=\lfloor n_i\rfloor+1}^{\infty}\frac{n_i^2}{k^2}
\le n_i + n_i^2\int_{n_i}^{\infty}x^{-2}dx = n_i + n_i = 2n_i.
\end{align*}
Hence, this sum is also bounded by a constant independent of $I$.
\\ \\
\noindent Putting these bounds together, we have shown that there exists a constant $C_A$ such that
\begin{align*}
\sum_{k=1}^{I-1}\Bigl[(G+2M\kappa)\E\bigl[\alpha_{S_k}\|\mu_{S_k}-\mu_\star\|\bigr] + 2M\kappa^2\E[\alpha_{S_k}^2]\Bigr] \le C_A.
\end{align*}
Consequently, $\E|A_I| \le C_A/I$.

\paragraph{3. Upper bound on $\E|m_I|$.}
Recall that $m_I = \frac{1}{I}\sum_{k=1}^{I-1}\E[\alpha_{S_k}]c_i - \Theta_i^I$ and $\Theta_i^I = \frac{n_i c_i}{\bar n}\frac{H_{I-1}}{I}$.
Thus,
\begin{align*}
|m_I| \le |c_i|\cdot \left|\frac{1}{I}\sum_{k=1}^{I-1}\E[\alpha_{S_k}] - \frac{n_i}{\bar n}\frac{H_{I-1}}{I}\right|.
\end{align*}
From Assumptions \ref{ass:bounded-feature} and~\ref{ass:smooth-utility}, we have $|c_i| \le \|g_\star\|\|\mu_i-\mu_\star\| \le G\cdot 2\kappa$.
It therefore suffices to bound the term in absolute value.

\subparagraph{3.1. Removing the $+n_i$ in the denominator.}
Since $\alpha_{S_k}=n_i/(N_{S_k}+n_i)$, define
\begin{align*}
    A_k
    :=
    \frac{n_i}{N_{S_k}+n_i}
    -
    \frac{n_i}{N_{S_k}} .
\end{align*}
For every realization of $S_k$, using $N_{S_k}\ge k$, we have
\begin{align*}
    |A_k|
    =
    \frac{n_i^2}{N_{S_k}(N_{S_k}+n_i)}
    \le
    \frac{n_i^2}{k^2}.
\end{align*}
Therefore,
\begin{align*}
\begin{aligned}
\left|
    \mathbb E[\alpha_{S_k}]
    -
    n_i\mathbb E\!\left[\frac{1}{N_{S_k}}\right]
\right|
&=
\left|
    \mathbb E[A_k]
\right| \\
&\le
\mathbb E[|A_k|] \\
&\le
\frac{n_i^2}{k^2}.
\end{aligned}
\end{align*}

\subparagraph{3.2. Comparing $\E[1/N_{S_k}]$ to $1/(k\bar n)$.}
Using the identity $|1/x-1/y| = |x-y|/(xy)$,
\begin{align*}
\left|\frac{1}{N_{S_k}} - \frac{1}{k\bar n}\right| = \frac{|N_{S_k}-k\bar n|}{N_{S_k}\,k\bar n} \le \frac{|N_{S_k}-k\bar n|}{k^2\bar n},
\end{align*}
since $N_{S_k}\ge k$. Taking expectations and applying Cauchy-Schwarz,
\begin{align*}
\E\left|\frac{1}{N_{S_k}} - \frac{1}{k\bar n}\right| \le \frac{\E|N_{S_k}-k\bar n|}{k^2\bar n}
\le \frac{\sqrt{\Var(N_{S_k})}}{k^2\bar n}.
\end{align*}
Now $N_{S_k} = \sum_{j\in S_k} n_j$, where the $n_j$ are i.i.d. with $\Var(n_j) \le n_{\max}^2$ (since $1\le n_j\le n_{\max}$). Hence $\Var(N_{S_k}) = k\Var(n_j) \le k n_{\max}^2$, so
\begin{align*}
\E\left|\frac{1}{N_{S_k}} - \frac{1}{k\bar n}\right| \le \frac{n_{\max}\sqrt{k}}{k^2\bar n} = \frac{n_{\max}}{\bar n}\frac{1}{k^{3/2}}.
\end{align*}
Multiplying by $n_i$ and combining with Step 3.1, we obtain
\begin{align*}
\left|\E[\alpha_{S_k}] - \frac{n_i}{k\bar n}\right| \le \frac{n_i^2}{k^2} + \frac{n_i n_{\max}}{\bar n}\frac{1}{k^{3/2}}.
\end{align*}

\subparagraph{3.3. Summation over $k$.}
The right-hand side is summable over $k$:
\begin{align*}
\sum_{k=1}^{\infty}\left(\frac{n_i^2}{k^2} + \frac{n_i n_{\max}}{\bar n}\frac{1}{k^{3/2}}\right) = n_i^2\sum_{k=1}^{\infty}\frac{1}{k^2} + \frac{n_i n_{\max}}{\bar n}\sum_{k=1}^{\infty}\frac{1}{k^{3/2}} < \infty.
\end{align*}
Denote this finite sum by $C_B$. Then
\begin{align*}
\sum_{k=1}^{I-1}\left|\E[\alpha_{S_k}] - \frac{n_i}{k\bar n}\right| \le C_B.
\end{align*}

\subparagraph{3.4. Final bound for $\E|m_I|$.}
We now have
\begin{align*}
\left|\frac{1}{I}\sum_{k=1}^{I-1}\E[\alpha_{S_k}] - \frac{n_i}{\bar n}\frac{H_{I-1}}{I}\right|
\le \frac{1}{I}\sum_{k=1}^{I-1}\left|\E[\alpha_{S_k}] - \frac{n_i}{k\bar n}\right|
\le \frac{C_B}{I}.
\end{align*}
Multiplying by $|c_i| \le 2G\kappa$ gives $\E|m_I| \le \frac{2G\kappa C_B}{I} =: \frac{C_B'}{I}$.

\paragraph{4. Conclusion.}
Combining the bounds for $\E|A_I|$, $\E|m_I|$, and the $k=0$ term, we obtain
\begin{align*}
\E|\phi_i^I - \Theta_i^I| \le \frac{C_A + C_B' + C_0}{I} = \frac{C_{\phi}}{I}.
\end{align*}
This completes the proof.
\end{proof}
\newpage

\subsection{Proof of Corollary \ref{cor:shapley-lower-bound}}
\label{subsec:proof-cor-shapley-lower-bound}
\begin{corollary}[Restate of Corollary \ref{cor:shapley-lower-bound}]
Under the assumptions of Theorem~\ref{thm:leading-term}, if $c_i\neq 0$, then $|\phi_i^I|=\Theta_{\mathbb P}\!\left(\frac{\log I}{I}\right).$
\end{corollary}
\begin{proof}
Let
\begin{align*}
    a_I:=\frac{\log I}{I}.
\end{align*}
By the definition of the leading term,
\begin{align*}
    \Theta_i^I
    =
    \frac{n_i c_i}{\bar n}
    \frac{H_{I-1}}{I}.
\end{align*}
Therefore,
\begin{align*}
    \frac{\Theta_i^I}{a_I}
    =
    \frac{n_i c_i}{\bar n}
    \frac{H_{I-1}}{\log I}.
\end{align*}
Since by Lemma \ref{lem:harmonic-asymptotics}, $H_{I-1}/\log I\to 1$, we have
\begin{align*}
    \frac{\Theta_i^I}{a_I}
    \longrightarrow
    \frac{n_i c_i}{\bar n}.
\end{align*}
Moreover, Theorem~\ref{thm:leading-term} gives
\begin{align*}
    \mathbb E
    \left[
        \left|
            \phi_i^I-\Theta_i^I
        \right|
    \right]
    =
    O\!\left(\frac{1}{I}\right).
\end{align*}
Hence,
\begin{align*}
\begin{aligned}
    \mathbb E
    \left[
        \left|
            \frac{\phi_i^I}{a_I}
            -
            \frac{\Theta_i^I}{a_I}
        \right|
    \right]
    &=
    \frac{1}{a_I}
    \mathbb E
    \left[
        \left|
            \phi_i^I-\Theta_i^I
        \right|
    \right]  \\
    &=
    \frac{I}{\log I}
    O\!\left(\frac{1}{I}\right)
    =
    O\!\left(\frac{1}{\log I}\right)
    \longrightarrow 0.
\end{aligned}
\end{align*}
Combining the last two displays yields
\begin{align*}
    \frac{I}{\log I}\phi_i^I
    =
    \frac{\phi_i^I}{a_I}
    \xlongrightarrow{L^1}
    \frac{n_i c_i}{\bar n}.
\end{align*}

If $c_i\neq0$, the limit is nonzero. Therefore,
\begin{align*}
    |\phi_i^I|
    =
    \Theta_{\mathbb P}(a_I)
    =
    \Theta_{\mathbb P}\!\left(\frac{\log I}{I}\right).
\end{align*}
\end{proof}

\newpage 
\subsection{Proof of Corollary \ref{cor:absolute-to-relative}}
\label{subsec:proof-cor-absolute-to-relative}
\begin{corollary}[Restate of Corollary \ref{cor:absolute-to-relative}]
Assume the conditions of Theorem~\ref{thm:leading-term} and suppose $c_i\neq 0$. Let $\widehat\phi_i^I$ be any estimator of $\phi_i^I$ such that $\left|\widehat\phi_i^I-\phi_i^I \right| = o_{\mathbb{P}}\!\left(\frac{\log I}{I}\right)$. Then, $\left|\frac{\widehat\phi_i^I}{\phi_i^I}-1\right|\xlongrightarrow{\mathbb P}0$.
\end{corollary}

\begin{proof}
Let
\begin{align*}
    a_I:=\frac{\log I}{I}.
\end{align*}
By Corollary~\ref{cor:shapley-lower-bound},
\begin{align*}
    \frac{\phi_i^I}{a_I}
    \xlongrightarrow{\mathbb P}
    \frac{n_i c_i}{\bar n}.
\end{align*}
Since $c_i\neq0$, the limit is nonzero. Moreover, by assumption,
\begin{align*}
    \frac{\widehat\phi_i^I-\phi_i^I}{a_I}
    \xlongrightarrow{\mathbb P}0.
\end{align*}
Therefore,
\begin{align*}
    \left|
        \frac{\widehat\phi_i^I}{\phi_i^I}
        -
        1
    \right|
    =
    \frac{
        \left|
            (\widehat\phi_i^I-\phi_i^I)/a_I
        \right|
    }{
        \left|
            \phi_i^I/a_I
        \right|
    }
    \xlongrightarrow{\mathbb P}0,
\end{align*}
by Slutsky's theorem.
\end{proof}

\newpage
\subsection{Proof of Proposition \ref{prop:mc_gt_budgets}}
\label{subsec:proof-prop-mc-gt}
\begin{proposition}[Restate of Proposition \ref{prop:mc_gt_budgets}]
Let $\widehat\varphi_{i,\mathrm{MC}}^I$ be the permutation Monte Carlo Shapley estimator \cite{maleki2013} based on $m_I$ independent random permutations. Let $\widehat\varphi_{i,\mathrm{GT}}^I$ be the group-testing Shapley estimator of \cite{jia2019b}. Under the assumptions of Theorem~\ref{thm:leading-term} and with $c_i \neq 0$:
\begin{enumerate}
\item If $m_I \gg \frac{I^2}{\log^2 I}$, then
\begin{align*}
\left|\frac{\widehat{\phi}_{i,\mathrm{MC}}^I}{\phi_i^I} - 1\right| \xlongrightarrow{\mathbb{P}} 0.
\end{align*}
\item If there exist sequences $\varepsilon_I=o(\log I/I)$ and $\delta_I\to0$ such that the group-testing budget $T_I$ is large enough for the estimator to be an
$(\varepsilon_I,\delta_I)$-approximation of the Shapley vector $(\phi_j^I)_{j=1}^{I}$ in $\ell_2$-norm., then
\begin{align*}
\left|\frac{\widehat{\phi}_{i,\mathrm{GT}}^I}{\phi_i^I} - 1\right| \xlongrightarrow{\mathbb{P}} 0.
\end{align*}
\end{enumerate}
\end{proposition}
\begin{proof}
We prove that both estimators have absolute error
$o_{\mathbb P}(\log I/I)$. The relative-error convergence then follows
directly from Corollary~\ref{cor:absolute-to-relative}.

Throughout the proof, write
\begin{align*}
    a_I:=\frac{\log I}{I}.
\end{align*}
We first record a uniform boundedness fact. By Assumption~\ref{ass:smooth-utility},
$\|\nabla F(x)\|\le G$, hence $F$ is $G$-Lipschitz. Moreover, by
Assumption~\ref{ass:bounded-feature}, every nonempty empirical embedding has
norm at most $\kappa$, and by convention
$\mu(\widehat P_\emptyset)=0$. Therefore, for any two coalitions $S,T$,
\begin{align*}
    |v(D_S)-v(D_T)|
    =
    \left|
        F(\mu(\widehat P_{D_S}))
        -
        F(\mu(\widehat P_{D_T}))
    \right|
    \le
    2G\kappa .
\end{align*}
Thus all marginal contributions are uniformly bounded:
\begin{align*}
    \left|
        v(D_S\biguplus D_i)-v(D_S)
    \right|
    \le
    2G\kappa .
    \tag{D.9}
    \label{eq:uniform-marginal-bound}
\end{align*}
Set
\begin{align*}
    r:=2G\kappa .
\end{align*}

\paragraph{Permutation Monte Carlo Shapley.}
Let $\pi_1,\dots,\pi_{m_I}$ be independent uniformly random permutations of
$[I]$. For a permutation $\pi$, let $P_i^\pi$ be the set of predecessors
of $i$ in $\pi$. The permutation Monte Carlo Shapley estimator is
\begin{align*}
    \widehat\varphi_{i,\mathrm{MC}}^I
    :=
    \frac{1}{m_I}
    \sum_{b=1}^{m_I}
    \left[
        v(D_{P_i^{\pi_b}}\biguplus D_i)
        -
        v(D_{P_i^{\pi_b}})
    \right].
\end{align*}
Conditionally on the realized $I$-player game, the summands are independent
and have mean $\varphi_i^I$, by the permutation representation of the Shapley
value. Moreover, by \eqref{eq:uniform-marginal-bound}, their absolute value is
bounded by $r$, so their variance is at most $r^2$. Hence, conditionally on
the realized game,
\begin{align*}
    \mathrm{Var}
    \left(
        \widehat\varphi_{i,\mathrm{MC}}^I
        \,\middle|\,
        (D_j)_{j\neq i}
    \right)
    \le
    \frac{r^2}{m_I}.
\end{align*}
By Chebyshev's inequality, for any $\eta>0$,
\begin{align*}
\begin{aligned}
    \mathbb P
    \left(
        \left|
            \widehat\varphi_{i,\mathrm{MC}}^I-\varphi_i^I
        \right|
        >
        \eta a_I
        \,\middle|\,
        (D_j)_{j\neq i}
    \right)
    &\le
    \frac{r^2}{m_I\eta^2 a_I^2}.
\end{aligned}
\end{align*}
Taking expectation over the surrounding players gives the same unconditional
bound:
\begin{align*}
    \mathbb P
    \left(
        \left|
            \widehat\varphi_{i,\mathrm{MC}}^I-\varphi_i^I
        \right|
        >
        \eta a_I
    \right)
    \le
    \frac{r^2}{m_I\eta^2 a_I^2}.
\end{align*}
If $m_Ia_I^2\to\infty$, then the right-hand side converges to zero for every
$\eta>0$. Therefore,
\begin{align*}
    \widehat\varphi_{i,\mathrm{MC}}^I-\varphi_i^I
    =
    o_{\mathbb P}(a_I)
    =
    o_{\mathbb P}\!\left(\frac{\log I}{I}\right).
\end{align*}
By Corollary~\ref{cor:absolute-to-relative},
\begin{align*}
    \left|
        \frac{\widehat\varphi_{i,\mathrm{MC}}^I}{\varphi_i^I}
        -1
    \right|
    \xrightarrow{\mathbb P}0.
\end{align*}

\paragraph{Group-testing Shapley.}
We now consider the group-testing estimator of \cite{jia2019b}, whose definition is recalled in Appendix \ref{sec:missing-defs}, applied to the
$I$-player dataset valuation game.

To match the standard convention $U(\emptyset)=0$, define the normalized game
\begin{align*}
    U_I(S):=v(D_S)-v(\emptyset),
    \qquad S\subseteq[I].
\end{align*}
This normalization does not change any marginal contribution, hence it does not
change the Shapley values. Moreover, by the Lipschitz bound above, the range of
$U_I$ is at most $r=2G\kappa$.

The group-testing estimator proceeds as follows. Let
\begin{align*}
    Z_I:=2\sum_{k=1}^{I-1}\frac{1}{k},
\end{align*}
and define, for $k=1,\dots,I-1$,
\begin{align*}
    q_I(k)
    :=
    \frac{1}{Z_I}
    \left(
        \frac{1}{k}
        +
        \frac{1}{I-k}
    \right).
\end{align*}
For each test $t=1,\dots,T_I$, draw $K_t\sim q_I$, then draw a subset
$A_t\subseteq[I]$ uniformly among subsets of cardinality $K_t$, and evaluate
$U_I(A_t)$. Let
\begin{align*}
    \beta_{t,j}:=\mathbf 1\{j\in A_t\}.
\end{align*}
For each pair $p,q\in[I]$, estimate the Shapley difference
$\varphi_p^I-\varphi_q^I$ by
\begin{align*}
    \widehat\Delta_{p,q}
    :=
    \frac{Z_I}{T_I}
    \sum_{t=1}^{T_I}
    U_I(A_t)(\beta_{t,p}-\beta_{t,q}).
\end{align*}
The estimated Shapley vector
$\widehat s^I=(\widehat s_1^I,\dots,\widehat s_I^I)$ is then obtained by
solving the feasibility problem
\begin{align*}
    \sum_{j=1}^I \widehat s_j^I
    =
    U_I([I]),
    \qquad
    \left|
        (\widehat s_p^I-\widehat s_q^I)
        -
        \widehat\Delta_{p,q}
    \right|
    \le
    \frac{\varepsilon_I}{2\sqrt I},
    \quad
    \forall p,q\in[I].
\end{align*}
The group-testing estimate of the fixed player is
\begin{align*}
    \widehat\varphi_{i,\mathrm{GT}}^I:=\widehat s_i^I.
\end{align*}

We now apply the guarantee of \cite{jia2019b}. Define
\begin{align*}
    q_{\mathrm{tot},I}
    :=
    \frac{I-2}{I}q_I(1)
    +
    \sum_{k=2}^{I-1}
    q_I(k)
    \left[
        1+
        \frac{2k(k-I)}{I(I-1)}
    \right],
\end{align*}
and
\begin{align*}
    h(u):=(1+u)\log(1+u)-u.
\end{align*}
Theorem~3 of \cite{jia2019b}, applied with $N=I$, states that if
\begin{align*}
    T_I
    \ge
    \frac{
        8\log\!\left(\frac{I(I-1)}{2\delta_I}\right)
    }{
        (1-q_{\mathrm{tot},I}^2)
        h\!\left(
            \frac{
                \varepsilon_I
            }{
                Z_I r\sqrt I(1-q_{\mathrm{tot},I}^2)
            }
        \right)
    },
    \tag{D.10}
    \label{eq:jia-gt-budget}
\end{align*}
then, conditionally on the realized game,
\begin{align*}
    \mathbb P
    \left(
        \left\|
            \widehat s^I-\varphi^I
        \right\|_2
        \le
        \varepsilon_I
        \,\middle|\,
        (D_j)_{j\neq i}
    \right)
    \ge
    1-\delta_I,
\end{align*}
where
\begin{align*}
    \varphi^I
    :=
    (\varphi_1^I,\dots,\varphi_I^I).
\end{align*}
The bound is uniform over games whose utility range is at most $r$, so taking
expectation over the surrounding players yields
\begin{align*}
    \mathbb P
    \left(
        \left\|
            \widehat s^I-\varphi^I
        \right\|_2
        >
        \varepsilon_I
    \right)
    \le
    \delta_I.
\end{align*}
In particular,
\begin{align*}
    \left|
        \widehat\varphi_{i,\mathrm{GT}}^I-\varphi_i^I
    \right|
    \le
    \left\|
        \widehat s^I-\varphi^I
    \right\|_2.
\end{align*}
Therefore, for any $\eta>0$,
\begin{align*}
\begin{aligned}
    \mathbb P
    \left(
        \left|
            \widehat\varphi_{i,\mathrm{GT}}^I-\varphi_i^I
        \right|
        >
        \eta a_I
    \right)
    &\le
    \mathbb P
    \left(
        \left\|
            \widehat s^I-\varphi^I
        \right\|_2
        >
        \eta a_I
    \right).
\end{aligned}
\end{align*}
Since $\varepsilon_I=o(a_I)$, for all sufficiently large $I$,
$\varepsilon_I\le \eta a_I$. Hence, for all sufficiently large $I$,
\begin{align*}
    \mathbb P
    \left(
        \left|
            \widehat\varphi_{i,\mathrm{GT}}^I-\varphi_i^I
        \right|
        >
        \eta a_I
    \right)
    \le
    \delta_I.
\end{align*}
Because $\delta_I\to0$, we obtain
\begin{align*}
    \widehat\varphi_{i,\mathrm{GT}}^I-\varphi_i^I
    =
    o_{\mathbb P}(a_I)
    =
    o_{\mathbb P}\!\left(\frac{\log I}{I}\right).
\end{align*}
By Corollary~\ref{cor:absolute-to-relative},
\begin{align*}
    \left|
        \frac{\widehat\varphi_{i,\mathrm{GT}}^I}{\varphi_i^I}
        -1
    \right|
    \xrightarrow{\mathbb P}0.
\end{align*}
This completes the proof.
\end{proof}

\newpage

\subsection{Proof of Theorem \ref{thm:leading_order_consistent}}
\label{subsec:proof-thm-leading-order-estimator}
\begin{theorem}[Restate of Theorem \ref{thm:leading_order_consistent}]
Under Assumptions~\ref{ass:population}--\ref{ass:smooth-utility} and with $c_i \neq 0$, every leading‑order estimator satisfies
\begin{align*}
\left|\frac{\widehat{\phi}_i^I}{\phi_i^I} - 1\right| \xlongrightarrow{\mathbb{P}} 0.
\end{align*}
\end{theorem}
\begin{proof}
From the representation of $\widehat{\phi}_i^I$, we have
\begin{align*}
\widehat{\phi}_i^I - \Theta_i^I = \frac{1}{I}\sum_{k=1}^{I-1}\frac{n_i}{k\bar\mu}\bigl(\widehat{\Delta}_{k,I}^{(i)} - c_i\bigr) + \frac{\widehat{\Delta}_{0,I}^{(i)}}{I}.
\end{align*}
Taking absolute values, expectations, and using the triangle inequality, we get
\begin{align*}
\mathbb{E}|\widehat{\phi}_i^I - \Theta_i^I|
\le \frac{n_i}{\bar\mu I}\sum_{k=1}^{I-1}\frac{1}{k}\mathbb{E}|\widehat{\Delta}_{k,I}^{(i)} - c_i|
   + \frac{\mathbb{E}|\widehat{\Delta}_{0,I}^{(i)}|}{I}
\le \frac{n_i}{\bar\mu I}C_{\Delta} + \frac{C_0}{I}
   = \frac{C}{I}.
\end{align*}
\end{proof}

\newpage
\subsection{Proof of Proposition \ref{prop:du_strat}}
\label{subsec:proof-prop-du-strat}
We first state two lemmas before proving Proposition \ref{prop:du_strat}.
\begin{lemma}[DU-Shapley summability bound]
\label{lem:du-close-leading}
Let $\widehat\phi_{i,\mathrm{DU}}^I$ denote the DU-Shapley estimator defined by
\begin{align*}
    \widehat\phi_{i,\mathrm{DU}}^I
    :=
    \frac{1}{I}
    \sum_{k=0}^{I-1}
    \widetilde\Delta_{k,I},
    \qquad
    \widetilde\Delta_{k,I}
    :=
    v(\widetilde D_{k,I}\uplus D_i)-v(\widetilde D_{k,I}).
\end{align*}
where, $\widetilde D_{k,I}$ is a multiset of $m_{k,I} := \left\lfloor k \frac{1}{I-1}\sum_{j\in [I]\setminus \{i\}} n_j\right\rfloor$
datapoints sampled uniformly without replacement from the pooled multiset $D_{[I]\setminus\{i\}}$. 
Under the assumptions of Theorem~\ref{thm:leading-term}, there exists a
constant $C_{\Delta,\mathrm{DU}}<\infty$, independent of $I$, such that
\begin{align}
    \sum_{k=1}^{I-1}
    \mathbb E
    \left[
        \left|
            \widetilde\Delta_{k,I}
            -
            \frac{n_i}{k\bar n}c_i
        \right|
    \right]
    \le
    C_{\Delta,\mathrm{DU}}.
    \label{eq:du-summability}
\end{align}

\end{lemma}

\begin{proof}
For $k\ge 1$, let
\begin{align*}
    \widetilde\mu_{k,I}
    :=
    \mu(\widehat P_{\widetilde D_{k,I}}),
    \qquad
    \widetilde\alpha_{k,I}
    :=
    \frac{n_i}{m_{k,I}+n_i},
    \qquad
    m_{k,I}:=\lfloor k\widehat n_{-i}^I\rfloor .
\end{align*}
Since $\widehat n_{-i}^I\ge 1$, we have $m_{k,I}\ge k$. Hence
\begin{align*}
    \widetilde\alpha_{k,I}
    \le
    \min\left\{1,\frac{n_i}{k}\right\},
    \qquad
    \widetilde\alpha_{k,I}^2
    \le
    \min\left\{1,\frac{n_i^2}{k^2}\right\}.
\end{align*}

We decompose
\begin{align*}
    \widetilde\Delta_{k,I}
    -
    \frac{n_i}{k\bar n}c_i
    =
    \left(
        \widetilde\Delta_{k,I}
        -
        \widetilde\alpha_{k,I}c_i
    \right)
    +
    \left(
        \widetilde\alpha_{k,I}
        -
        \frac{n_i}{k\bar n}
    \right)c_i.
\end{align*}
Therefore,
\begin{align*}
\begin{aligned}
    \mathbb E
    \left[
        \left|
            \widetilde\Delta_{k,I}
            -
            \frac{n_i}{k\bar n}c_i
        \right|
    \right]
    &\le
    \mathbb E
    \left[
        \left|
            \widetilde\Delta_{k,I}
            -
            \widetilde\alpha_{k,I}c_i
        \right|
    \right]
    +
    |c_i|
    \mathbb E
    \left[
        \left|
            \widetilde\alpha_{k,I}
            -
            \frac{n_i}{k\bar n}
        \right|
    \right].
\end{aligned}
\end{align*}

For the first term, applying the same first-order marginal bound as
Lemma~\ref{lem:marginal_bound}, with
$(\mu_S,\alpha_S)$ replaced by
$(\widetilde\mu_{k,I},\widetilde\alpha_{k,I})$, gives
\begin{align*}
    \left|
        \widetilde\Delta_{k,I}
        -
        \widetilde\alpha_{k,I}c_i
    \right|
    \le
    \widetilde\alpha_{k,I}(G+2M\kappa)
    \|\widetilde\mu_{k,I}-\mu^\star\|
    +
    2M\kappa^2\widetilde\alpha_{k,I}^2.
\end{align*}
Moreover, by the triangle inequality,
\begin{align*}
    \|\widetilde\mu_{k,I}-\mu^\star\|
    \le
    \|\widetilde\mu_{k,I}-\mu_{-i}^I\|
    +
    \|\mu_{-i}^I-\mu^\star\|,
\end{align*}
where $\mu_{-i}^I:=\mu(\widehat P_{D_{-i}^I})$. By the
without-replacement concentration corollary and $m_{k,I}\ge k$,
\begin{align*}
    \mathbb E\|\widetilde\mu_{k,I}-\mu_{-i}^I\|
    \le
    \kappa\sqrt{\frac{8\pi}{k}}.
\end{align*}
By Lemma~\ref{lem:coalition-embedding-concentration} applied to the coalition
$[I]\setminus\{i\}$,
\begin{align*}
    \mathbb E\|\mu_{-i}^I-\mu^\star\|
    \le
    \frac{2n_{\max}\kappa}{\sqrt{I-1}}.
\end{align*}
Thus, for constants $C_1,C_2<\infty$,
\begin{align*}
    \mathbb E\|\widetilde\mu_{k,I}-\mu^\star\|
    \le
    \frac{C_1}{\sqrt{k}}
    +
    \frac{C_2}{\sqrt{I-1}}.
\end{align*}
Consequently,
\begin{align*}
\begin{aligned}
\mathbb E
\left[
    \left|
        \widetilde\Delta_{k,I}
        -
        \widetilde\alpha_{k,I}c_i
    \right|
\right]
&\le
C
\min\left\{1,\frac{n_i}{k}\right\}
\left(
    \frac{1}{\sqrt{k}}
    +
    \frac{1}{\sqrt{I-1}}
\right)
+
C
\min\left\{1,\frac{n_i^2}{k^2}\right\}.
\end{aligned}
\end{align*}
Summing over $k$ yields
\begin{align*}
    \sum_{k=1}^{I-1}
    \mathbb E
    \left[
        \left|
            \widetilde\Delta_{k,I}
            -
            \widetilde\alpha_{k,I}c_i
        \right|
    \right]
    \le
    \widetilde C_A,
\end{align*}
because the sums
\begin{align*}
    \sum_{k\ge1}
    \min\left\{1,\frac{n_i}{k}\right\}k^{-1/2},
    \qquad
    \sum_{k\ge1}
    \min\left\{1,\frac{n_i^2}{k^2}\right\}
\end{align*}
are finite, and
\begin{align*}
    \frac{1}{\sqrt{I-1}}
    \sum_{k=1}^{I-1}
    \min\left\{1,\frac{n_i}{k}\right\}
    \le
    C\frac{1+\log I}{\sqrt{I-1}}
    \le C.
\end{align*}

For the second term, since
\begin{align*}
    \left|
        \widetilde\alpha_{k,I}
        -
        \frac{n_i}{k\bar n}
    \right|
    \le
    \left|
        \frac{n_i}{m_{k,I}+n_i}
        -
        \frac{n_i}{k\widehat n_{-i}^I+n_i}
    \right|
    +
    \left|
        \frac{n_i}{k\widehat n_{-i}^I+n_i}
        -
        \frac{n_i}{k\bar n+n_i}
    \right|
    +
    \left|
        \frac{n_i}{k\bar n+n_i}
        -
        \frac{n_i}{k\bar n}
    \right|,
\end{align*}
using $|m_{k,I}-k\widehat n_{-i}^I|\le 1$, $m_{k,I}\ge k$, and
$\widehat n_{-i}^I\ge 1$, we get
\begin{align*}
    \mathbb E
    \left[
        \left|
            \widetilde\alpha_{k,I}
            -
            \frac{n_i}{k\bar n}
        \right|
    \right]
    \le
    \frac{C}{k^2}
    +
    \frac{C}{k\sqrt{I-1}}.
\end{align*}
Indeed,
\begin{align*}
    \mathbb E|\widehat n_{-i}^I-\bar n|
    \le
    \sqrt{\operatorname{Var}(\widehat n_{-i}^I)}
    \le
    \frac{n_{\max}}{\sqrt{I-1}}.
\end{align*}
Hence
\begin{align*}
    \sum_{k=1}^{I-1}
    \mathbb E
    \left[
        \left|
            \widetilde\alpha_{k,I}
            -
            \frac{n_i}{k\bar n}
        \right|
    \right]
    \le
    \widetilde C_B
\end{align*}
for a constant $\widetilde C_B<\infty$ independent of $I$. Since
\begin{align*}
    |c_i|
    =
    |\langle\nabla F(\mu^\star),\mu_i-\mu^\star\rangle|
    \le
    2G\kappa,
\end{align*}
we obtain
\begin{align*}
    \sum_{k=1}^{I-1}
    \mathbb E
    \left[
        \left|
            \widetilde\Delta_{k,I}
            -
            \frac{n_i}{k\bar n}c_i
        \right|
    \right]
    \le
    \widetilde C_A+2G\kappa\,\widetilde C_B
    =:
    C_{\Delta,\mathrm{DU}}.
\end{align*}
This proves \eqref{eq:du-summability}.
\end{proof}
\begin{lemma}[Stratified Monte Carlo summability bound]
\label{lem:strat-close-leading}
Let $\widehat\phi_{i,\mathrm{strat}}^I$ be the stratified Monte Carlo Shapley estimator defined by
\begin{align*}
    \widehat\phi_{i,\mathrm{strat}}^I
    :=
    \frac{1}{I}
    \sum_{k=0}^{I-1}
    \frac{1}{m_{k,I}}
    \sum_{b=1}^{m_{k,I}}
    \Delta_i(S_{k,b}),
\end{align*}
where $m_{k,I}\ge 1$ and $S_{k,b}$ is sampled uniformly among subsets of
$[I]\setminus\{i\}$ of size $k$. Under the assumptions of
Theorem~\ref{thm:leading-term}, there exists a constant
$C_{\Delta,\mathrm{strat}}<\infty$, independent of $I$, such that
\begin{align}
    \sum_{k=1}^{I-1}
    \mathbb E
    \left[
        \left|
            \frac{1}{m_{k,I}}
            \sum_{b=1}^{m_{k,I}}
            \Delta_i(S_{k,b})
            -
            \frac{n_i}{k\bar n}c_i
        \right|
    \right]
    \le
    C_{\Delta,\mathrm{strat}}.
    \label{eq:strat-summability}
\end{align}
\end{lemma}

\begin{proof}
By Jensen's inequality,
\begin{align*}
\begin{aligned}
&\mathbb E
\left[
    \left|
        \frac{1}{m_{k,I}}
        \sum_{b=1}^{m_{k,I}}
        \Delta_i(S_{k,b})
        -
        \frac{n_i}{k\bar n}c_i
    \right|
\right] \\
&\qquad\le
\frac{1}{m_{k,I}}
\sum_{b=1}^{m_{k,I}}
\mathbb E
\left[
    \left|
        \Delta_i(S_{k,b})
        -
        \frac{n_i}{k\bar n}c_i
    \right|
\right] \\
&\qquad=
\mathbb E_{S_k}
\left[
    \left|
        \Delta_i(S_k)
        -
        \frac{n_i}{k\bar n}c_i
    \right|
\right].
\end{aligned}
\end{align*}
Let
\begin{align*}
    \alpha_{S_k}:=\frac{n_i}{N_{S_k}+n_i}.
\end{align*}
Then
\begin{align*}
    \Delta_i(S_k)-\frac{n_i}{k\bar n}c_i
    =
    \bigl(\Delta_i(S_k)-\alpha_{S_k}c_i\bigr)
    +
    \left(
        \alpha_{S_k}-\frac{n_i}{k\bar n}
    \right)c_i.
\end{align*}
By Lemma~\ref{lem:marginal_bound},
\begin{align*}
    |\Delta_i(S_k)-\alpha_{S_k}c_i|
    \le
    \alpha_{S_k}(G+2M\kappa)
    \|\mu_{S_k}-\mu^\star\|
    +
    2M\kappa^2\alpha_{S_k}^2.
\end{align*}
Using Lemmas~\ref{lem:alpha_bound} and
\ref{lem:coalition-embedding-concentration}, the same summability argument as
in the proof of Theorem~\ref{thm:leading-term} gives
\begin{align*}
    \sum_{k=1}^{I-1}
    \mathbb E_{S_k}
    \left[
        \left|
            \Delta_i(S_k)-\alpha_{S_k}c_i
        \right|
    \right]
    \le
    C_A.
\end{align*}
Moreover,
\begin{align*}
    \mathbb E_{S_k}
    \left[
        \left|
            \alpha_{S_k}-\frac{n_i}{k\bar n}
        \right|
    \right]
    \le
    \frac{n_i^2}{k^2}
    +
    \frac{n_i n_{\max}}{\bar n\,k^{3/2}}.
\end{align*}
Indeed, this follows by writing
\begin{align*}
\left|
    \alpha_{S_k}-\frac{n_i}{k\bar n}
\right|
\le
\left|
    \frac{n_i}{N_{S_k}+n_i}
    -
    \frac{n_i}{N_{S_k}}
\right|
+
\left|
    \frac{n_i}{N_{S_k}}
    -
    \frac{n_i}{k\bar n}
\right|,
\end{align*}
using $N_{S_k}\ge k$, and applying
\begin{align*}
    \mathbb E_{S_k}|N_{S_k}-k\bar n|
    \le
    \sqrt{\operatorname{Var}(N_{S_k})}
    \le
    n_{\max}\sqrt{k}.
\end{align*}
Therefore,
\begin{align*}
    \sum_{k=1}^{I-1}
    \mathbb E_{S_k}
    \left[
        \left|
            \alpha_{S_k}-\frac{n_i}{k\bar n}
        \right|
    \right]
    \le
    C_B
\end{align*}
for a constant $C_B<\infty$ independent of $I$. Since
$|c_i|\le 2G\kappa$, we obtain
\begin{align*}
    \sum_{k=1}^{I-1}
    \mathbb E
    \left[
        \left|
            \frac{1}{m_{k,I}}
            \sum_{b=1}^{m_{k,I}}
            \Delta_i(S_{k,b})
            -
            \frac{n_i}{k\bar n}c_i
        \right|
    \right]
    \le
    C_A+2G\kappa C_B
    =:
    C_{\Delta,\mathrm{strat}}.
\end{align*}
This proves \eqref{eq:strat-summability}.
\end{proof}
\begin{proposition}[Restate of Proposition \ref{prop:du_strat}]
Let $\widehat{\phi}_{i,\mathrm{DU}}^I$ denote the DU‑Shapley estimator \citep{garrido2024}.  
Let $\widehat{\phi}_{i,\mathrm{strat}}^I$ denote the stratified Monte Carlo Shapley estimator \cite{maleki2013}. Then, under the assumptions of Theorem~\ref{thm:leading-term}, both $\widehat{\phi}_{i,\mathrm{DU}}^I$ and $\widehat{\phi}_{i,\mathrm{strat}}^I$ are leading‑order estimators. Consequently,
\begin{align*}
\left|\frac{\widehat{\phi}_{i,\mathrm{DU}}^I}{\phi_i^I} - 1\right| \xlongrightarrow{\mathbb{P}} 0,\qquad
\left|\frac{\widehat{\phi}_{i,\mathrm{strat}}^I}{\phi_i^I} - 1\right| \xlongrightarrow{\mathbb{P}} 0.
\end{align*}
\end{proposition}
\begin{proof}
We prove the statement for both estimators at once.

Let $E\in\{\mathrm{DU},\mathrm{strat}\}$. For each estimator, write
\begin{align*}
    \widehat\phi_{i,E}^I
    =
    \frac{1}{I}
    \sum_{k=0}^{I-1}
    \Gamma_{k,I}^{E},
\end{align*}
where
\begin{align*}
    \Gamma_{k,I}^{\mathrm{DU}}
    :=
    \widetilde\Delta_{k,I}
    =
    v(\widetilde D_{k,I}\uplus D_i)-v(\widetilde D_{k,I}),
\end{align*}
and
\begin{align*}
    \Gamma_{k,I}^{\mathrm{strat}}
    :=
    \frac{1}{m_{k,I}}
    \sum_{b=1}^{m_{k,I}}
    \Delta_i(S_{k,b}).
\end{align*}
For $k=0$, both estimators use the empty-coalition marginal contribution,
\begin{align*}
    \Gamma_{0,I}^{E}
    =
    v(D_i)-v(\emptyset)
    =
    F(\mu_i)-F(0).
\end{align*}
For $k\ge 1$, define
\begin{align*}
    \widehat\Delta_{k,I}^{(i,E)}
    :=
    \frac{k\bar n}{n_i}\Gamma_{k,I}^{E},
    \qquad
    \widehat\Delta_{0,I}^{(i,E)}
    :=
    \Gamma_{0,I}^{E}.
\end{align*}
Then
\begin{align*}
    \widehat\phi_{i,E}^I
    =
    \frac{1}{I}
    \sum_{k=1}^{I-1}
    \frac{n_i}{k\bar n}
    \widehat\Delta_{k,I}^{(i,E)}
    +
    \frac{\widehat\Delta_{0,I}^{(i,E)}}{I}.
\end{align*}
Thus, both estimators have the structural form required in
Definition~\ref{def:leading_order}. It remains to verify whether conditions (C1) and (C2) hold.

\paragraph{Condition (C1).}
Since $F$ is $G$-Lipschitz and $\|\mu_i\|\le \kappa$,
\begin{align*}
    \left|
        \widehat\Delta_{0,I}^{(i,E)}
    \right|
    =
    |F(\mu_i)-F(0)|
    \le
    G\kappa.
\end{align*}
Therefore,
\begin{align*}
    \mathbb E
    \left[
        \left|
            \widehat\Delta_{0,I}^{(i,E)}
        \right|
    \right]
    \le
    G\kappa,
\end{align*}
uniformly in $I$, for both $E=\mathrm{DU}$ and
$E=\mathrm{strat}$.

\paragraph{Condition (C2).}
Let
\begin{align*}
    \beta_k:=\frac{n_i}{k\bar n}.
\end{align*}
For $k\ge 1$,
\begin{align*}
    \widehat\Delta_{k,I}^{(i,E)}-c_i
    =
    \frac{k\bar n}{n_i}
    \left(
        \Gamma_{k,I}^{E}-\beta_k c_i
    \right).
\end{align*}
Hence
\begin{align*}
    \frac{1}{k}
    \mathbb E
    \left[
        \left|
            \widehat\Delta_{k,I}^{(i,E)}-c_i
        \right|
    \right]
    =
    \frac{\bar n}{n_i}
    \mathbb E
    \left[
        \left|
            \Gamma_{k,I}^{E}-\beta_k c_i
        \right|
    \right].
\end{align*}
Summing over $k$, we obtain
\begin{align*}
    \sum_{k=1}^{I-1}
    \frac{1}{k}
    \mathbb E
    \left[
        \left|
            \widehat\Delta_{k,I}^{(i,E)}-c_i
        \right|
    \right]
    =
    \frac{\bar n}{n_i}
    \sum_{k=1}^{I-1}
    \mathbb E
    \left[
        \left|
            \Gamma_{k,I}^{E}-\frac{n_i}{k\bar n}c_i
        \right|
    \right].
\end{align*}

We now use the summability estimates proved in Lemmas~\ref{lem:du-close-leading}
and~\ref{lem:strat-close-leading}. More precisely, 
Lemma~\ref{lem:du-close-leading} establishes that there exists a constant
$C_{\Delta,\mathrm{DU}}<\infty$, independent of $I$, such that
\begin{align*}
    \sum_{k=1}^{I-1}
    \mathbb E
    \left[
        \left|
            \widetilde\Delta_{k,I}
            -
            \frac{n_i}{k\bar n}c_i
        \right|
    \right]
    \le
    C_{\Delta,\mathrm{DU}}.
\end{align*}
Likewise, Lemma~\ref{lem:strat-close-leading} establishes that
there exists a constant $C_{\Delta,\mathrm{strat}}<\infty$, independent of
$I$, such that
\begin{align*}
    \sum_{k=1}^{I-1}
    \mathbb E
    \left[
        \left|
            \frac{1}{m_{k,I}}
            \sum_{b=1}^{m_{k,I}}
            \Delta_i(S_{k,b})
            -
            \frac{n_i}{k\bar n}c_i
        \right|
    \right]
    \le
    C_{\Delta,\mathrm{strat}}.
\end{align*}
Therefore, for $E\in\{\mathrm{DU},\mathrm{strat}\}$,
\begin{align*}
    \sum_{k=1}^{I-1}
    \frac{1}{k}
    \mathbb E
    \left[
        \left|
            \widehat\Delta_{k,I}^{(i,E)}-c_i
        \right|
    \right]
    \le
    \frac{\bar n}{n_i}C_{\Delta,E},
\end{align*}
where $C_{\Delta,E}$ denotes the corresponding constant. This verifies
condition (C2) for both estimators.

Thus both $\widehat\phi_{i,\mathrm{DU}}^I$ and
$\widehat\phi_{i,\mathrm{strat}}^I$ are leading-order estimators.

Finally, by Theorem~\ref{thm:leading_order_consistent}, every leading-order estimator is relatively consistent whenever $c_i\neq 0$. Hence
\begin{align*}
    \left|
        \frac{\widehat\phi_{i,\mathrm{DU}}^I}{\phi_i^I}
        -1
    \right|
    \xrightarrow{\mathbb P}0,
    \qquad
    \left|
        \frac{\widehat\phi_{i,\mathrm{strat}}^I}{\phi_i^I}
        -1
    \right|
    \xrightarrow{\mathbb P}0.
\end{align*}
\end{proof}

\end{document}